\documentclass[a4-paper,pdflatex,sn-mathphys-num]{sn-jnl}

\usepackage{lmodern}

\usepackage{bbm}
\usepackage[utf8]{inputenc}
\usepackage[english]{babel}
\usepackage{csquotes}       
\usepackage{microtype}      
\usepackage{amssymb}        
\usepackage{mathtools}      
\usepackage{physics}
\usepackage{xcolor}

\usepackage{dirtytalk}      

\usepackage{amssymb}
\usepackage{soul}

\usepackage{braket}         

\colorlet{mutedblue}{blue!70!black}
\colorlet{mutedred}{red!70!black}
\colorlet{mutedgreen}{green!70!black}

\usepackage{mathrsfs}

\numberwithin{equation}{section}

\newcommand{\one}{\mathbbm{1}}
\newcommand{\N}{\mathbb N}
\newcommand{\Z}{\mathbb Z}

\newcommand{\R}{\mathbb R}

\newcommand{\C}{\mathbb C}
\newcommand{\cA}{\mathcal A}
\newcommand{\cB}{\mathcal B}
\newcommand{\cC}{\mathcal C}
\newcommand{\cD}{\mathcal D}

\newcommand{\cF}{\mathcal F}

\newcommand{\cH}{\mathcal H}

\newcommand{\cK}{\mathcal K}

\newcommand{\cR}{\mathcal R}
\newcommand{\cS}{\mathcal S}

\newcommand{\cW}{\mathcal W}

\newcommand{\RDM}{\mathrm{RDM}}

\newcommand{\CE}{\mathrm{CE}}

\def\epsilon{\varepsilon}

\DeclareMathOperator{\supp}{supp}

\DeclareMathOperator{\dom}{dom}

\DeclareMathOperator{\ri}{ri}
\DeclareMathOperator{\Aff}{Aff}
\DeclareMathOperator{\Conv}{Conv}
\DeclareMathOperator{\ran}{ran}

\usepackage{amsthm}
\theoremstyle{plain}
\newtheorem{theorem}{Theorem}[section]
\newtheorem{proposition}[theorem]{Proposition}
\newtheorem{lemma}[theorem]{Lemma}
\newtheorem{corollary}[theorem]{Corollary}

\newtheorem*{lemma*}{Lemma}

\theoremstyle{definition}
\newtheorem{definition}[theorem]{Definition}

\theoremstyle{remark}
\newtheorem{example}[theorem]{Example}
\newtheorem{remark}[theorem]{Remark}

\date{\today}

\author[$*$1]{\fnm{Håkon R.} \sur{Fredheim}}\email{haakorfr@uio.no}

\author[1]{\fnm{Simen} \sur{Kvaal}}\email{simen.kvaal@kjemi.uio.no}

\affil[1]{\orgdiv{Hylleraas Centre for Quantum Molecular Sciences}, 
          \orgdiv{Department of Chemistry}, 
          \orgname{University of Oslo}, 
          \orgaddress{\street{P.O. box 1053, Blindern}, 
                      \postcode{0316}, 
                      \city{Oslo}, 
                       \country{Norway}}}

\title[]{Reduced Density Matrix Functional Theory and a Reduced Formulation of Density Functional Theory}

\abstract{%
A mathematical framework for reduced density matrix functional theory (RDMFT) is proposed. The work is inspired by and generalizes the work by E.H.~Lieb [E.H.~Lieb, \emph{Int.~J.~Quant.~Chem.} 24(1983), pp.~243--277] on density-functional theory (DFT). We introduce a Banach space for density matrices with finite kinetic energy. The dual space is a rich class of single-particle potentials, i.e., Hermitian forms. The ground state energy of an $N$-fermion system with external forces given by any such Hermitian form is expressed as the Legendre--Fenchel transform of a convex and lower semicontinuous ``universal'' reduced density matrix functional. The formalism is employed to provide a mathematical framework for density-functional theory (DFT). The main tool here is a rigorous definition of diagonals of reduced density matrices. The result is a refinement of Lieb's results on DFT applicable to a wide variety of models.
}

\DeclareRobustCommand{\SkipTocEntry}[5]{}

\begin{document}

\makeatletter
\@namedef{subjclassname@2020}{\textup{2020} Mathematics Subject Classification}
\makeatother

\maketitle


\tableofcontents

\setcounter{tocdepth}{1}

\section*{Acknowledgments}

We would like to thank André Laestadius, Nadia Larsen and Sergey Neshveyev for engaging and insightful discussions. Furthermore, we would like to thank Hannes Wendt, who introduced us to the $X$-norm and helped in the formulation and proof of Theorem \ref{Theorem: Diagonal map}. We would also like to thank Thiago Corso, who pointed out out a flaw in our proof of Theorem \ref{Theorem: Surjective diagonal map} and contributed by formulating an alternative proof, including Lemma \ref{Lemma: Positive increasing sequence}. The work was supported by the Research Council
of Norway Project Nos.~262695 (Hylleraas Centre for Quantum Molecular Sciences), and 324944 (QOMBINE).

\section{Introduction}

Density-functional theory (DFT) is the computational workhorse of quantum chemistry and solid state physics~\cite{hohenberg_inhomogeneous_1964,kohn_self-consistent_1965,teale_dft_2022}. Its tremendous success stems from the fact that it models an $N$-fermion system using the single-particle density instead of the full wavefunction. This reduction is motivated by the famous Hohenberg--Kohn theorem~\cite{hohenberg_inhomogeneous_1964,garrigue_unique_2018} which states that for a molecular system, the single-electron density $\rho \in L^1_\R(\R^3)$ determines the $N$-electron Hamiltonian up to a constant shift, and therefore all physical observables of an $N$-electron system. A mathematical foundation for DFT was provided by E.H.~Lieb in 1983~\cite{Lieb1983}. Among the important results is the identification of the ground-state energy $E : L^{3/2}_\R(\R^3) + L^\infty_\R(\R^3)\to\R$ of an $N$-electron system as function of the external potential field $v$ with a Legendre--Fenchel transform of a  ``universal'' functional $F : L^1_\R(\R^3)\cap L^3_\R(\R^3) \to \R\cup\{+\infty\}$, that sends a density $\rho$ to its intrinsic energy. The functional $F$ is the mathematical realization of the universal functional introduced by Hohenberg and Kohn,  the subject of now-standard approximations used in Kohn--Sham DFT~\cite{kohn_self-consistent_1965,toulouse_review_2023}. We remark here, that the density space $L^1_\R(\R^3)\cap L^3_\R(\R^3)$ was somewhat arbitrarily chosen so it contains all densities with finite electronic kinetic energy, while the dual space $L^{3/2}_\R(\R^3)+L^\infty_\R(\R^3)$, contains elements of physical relevance, e.g., the atomic Coulomb potential.

Reduced (one-particle) density matrix functional theory (RDMFT, or 1-RDMFT) is a generalization of DFT that allows for non-local external potential operators~\cite{gilbertHohenbergKohnTheoremNonlocal1975a,cioslowskiManyElectronDensitiesReduced2000,Christian2018,pernal_reduced_2015}, and therefore one needs to consider the reduced density matrix $\gamma$, a positive trace-class operator over $L^2(\R^3\times\Z^2)$, as the variable of the universal functional. Although formally, the ground state energy and the RDM-functional are known to form a convex-conjugate pair~\cite{liebertRefiningRelatingFundamentals2023}, a rigorous mathematical treatment that handles unbounded Hamiltomians or infinite dimensional state spaces seems to be missing. It is therefore the mission of this publication to provide said mathematical treatment, general enough to be used by most, if not all realizations of RDMFT. The main job is to introduce a proper Banach space $X$ of reduced density matrices with finite kinetic energy. Once constructed, the dual $X^*$ consists of certain Hermitian sesquilinear forms. Some of these forms, denoted $\cR \subset X^*$, are associated with lower semibounded self-adjoint one-particle Hamiltonians, including those that contain kinetic energy plus Coulomb potentials. A universal interaction functional $F_\RDM : X \to \R\cup\{+\infty\}$ is constructed by means of the usual constrained-search procedure, and it is demonstrated to be convex lower-semicontinuous. The Legendre transform is the ground-state energy $E_\RDM : X^* \to \R \cup \{-\infty\}$.

As the diagonal of a $\gamma$ is a density $\rho$, RDMFT naturally restricts to DFT. Using projection-valued measures, we construct an abstract characterization of what a diagonal of a density operator should be. The result is an abstract theory for a wide range of DFTs. In particular, one obtains a refinement of Lieb's analysis for the molecular DFT, including a Banach space $\Xi \subset L^1_\R(\R^3)\cap L^3_\R(\R^3)$, which, unlike the latter space, contains \emph{no} elements of infinite kinetic energy. One can therefore claim it is better suited for the analysis of DFT. The dual space $\Xi^*$ contains the Kato-Rellich class of potentials $  L^{3/2}_\R(\R^3) + L^\infty_\R(\R^3)$, accounted for by Lieb's DFT, as well as more general classes like the Rollnik class of potentials~\cite{Simon1971}.

Our results should have some relevance for other fields of study in many-fermion theory, such as the rigorous formulation of Kohn--Sham theory, the theory of strictly correlated DFT~\cite{Seidl1999,Seidl2017}, and two-particle reduced density matrix functional theory~\cite{cioslowskiManyElectronDensitiesReduced2000,CancesLewin2006}.

The article is organized as follows: In Section~\ref{Section: Preliminaries} we introduce some important notation and notions. In Section~\ref{Section: Reduced Density Matrices}, we introduce $X$, the Banach space of reduced density matrices with finite kinetic energy, while in Section~\ref{Section: Single-particle Hamiltonians}, we discuss the dual space $X^*$. In Section~\ref{Section: Surjectivity of the partial trace map}, we study the partial trace map from $X^N$, the Banach space of $N$-fermion density operators, to $X$ and show that it is surjective. In Section~\ref{Section: Reduced Density Matrix Functional Theory}, we set up RDMFT using the formalism developed in the preceding sections. Finally, in Section~\ref{Section: Relation to DFT}, we introduce the diagonal map and set up a formalism for a general set of DFTs. An Appendix contains the proof of a lemma.

\section{Preliminaries}\label{Section: Preliminaries}

\subsection{{\itshape N}-fermion spaces and operators}

In this paper, $\cH$ shall denote a complex separable Hilbert space. The inner product $\langle\,\cdot\,,\,\cdot\,\rangle$ is linear in the second argument and antilinear in the first argument, and complex conjugation is denoted by an overline $\overline{z}$. We shall use the notation  
\[\cH^N := \bigwedge^N\cH\] 
for the $N$-fold antisymmetric tensor product of $\cH$ with itself. For $\varphi_1 ,\cdots ,\varphi_N\in \cH$, we denote the corresponding normalized Slater determinant by $\varphi_1\wedge\cdots \wedge\varphi_N\in \bigwedge^N\cH$. For $\varphi,\psi\in \cH$, we denote by $|\varphi\rangle\langle\psi|$ the rank-one operator given by 
\[|\varphi\rangle \langle \psi|\xi := \varphi\langle \psi,\xi\rangle \] for all $\xi\in \cH$. 

For an operator $A$ on $\cH$, the second quantization $\hat A$ is the operator on $\bigwedge^N\cH$ acting on Slater determinants as
\begin{equation}\label{Equation: Second quantization}
\hat A(\varphi_1\wedge\cdots\wedge\varphi_N)=\sum_{i=1}^N\varphi_1\wedge\cdots \wedge (A\varphi_i)\wedge\cdots \wedge\varphi_N.
\end{equation}

If $T$ is a positive unbounded self-adjoint operator, then there is a ``canonical'' way of defining the second quantization $\hat T$ as a self-adjoint operator on $\bigwedge^N\cH$ satisfying \[\hat T(\varphi_1\wedge\cdots\wedge\varphi_N)=\sum_{i=1}^N\varphi_1\wedge\cdots \wedge (T\varphi_i)\wedge\cdots \wedge\varphi_N\]
if $\varphi_1,\cdots,\varphi_N\in D(T)$. This is well known, and the interested reader can consult, e.g., Ref.~\cite[Chapter 5]{BratteliVol2} for an exact definition of $\hat T$. 

In a Banach space $E$ we denote the open ball of radius $r$ centered at $x$ by
\[B_{\|\cdot\|_E}(x,r):= \{ x'\in E \mid \|x'-x\|_E< r\},\]
where the dependence on the choice of norm is made explicit to avoid confusion. The corresponding closed ball will be denoted $\overline{B}_{\|\cdot\|_E}(x,r)$.

\subsection{Trace-class operators}
See, e.g., Ref.~\cite{reed_methods_1980} for a comprehensive treatment of bounded operators.
In this article, $\cB(\cH)$ denotes the Banach space of bounded operators on $\cH$ and the operator norm is denoted by $\|\cdot\|$, the same notation as for the norm on $\cH$. We denote the set of trace-class operators on $\cH$ by 
 \[\cS_1(\cH):= \{A\in \cB(\cH)\mid \Tr(|A|)<+\infty\},\]
and the set of compact operators on $\cH$ by $\cK(\cH)$. $\cS_1(\cH)$ is itself a Banach space when equipped with the norm $\|A\|_1 := \Tr(|A|)$, and its dual space is $\cS_1(\cH)^*=\cB(\cH)$. Its predual is given by $\cS_1(\cH)_* = \cK(\cH)$. 

By ``$V$ is a positive operator'', we shall mean that $V$ is \emph{positive semidefinite and self-adjoint}, written $V \geq 0$ for brevity. A possibly unbounded operator $V$ is positive semidefinite if and only if $\braket{\psi,V\psi} \geq 0$ for every $\psi\in D(V)$. If the inequality is strict, we say that $V$ is strictly positive or positive definite, written $V>0$.  Similarly, we also write $V \geq V'$ ($V > V'$) whenever $\braket{\psi,V\psi}\geq (>) \braket{\psi,V'\psi}$ on $D(V)\cap D(V')$, but $V-V'$ may happen to be non-self-adjoint or not even densely defined.

If $\cA\subset\cB(\cH)$, we denote the subset of self-adjoint operators by $\cA^{\mathrm{sa}} \subset \cA$, the subset of positive operators by $\cA^+$, e.g., $\cS^{\mathrm{sa}}_1(\cH)$ and $\cK^+(\cH)$. In particular, the subset of self-adjoint trace-class operators $\cS_1^{\text{sa}}(\cH)$ is a real Banach space under the trace norm.

\begin{lemma}\label{lemma: gauge characterization of trace norm}
    The trace norm as defined on $\cS_1^\mathrm{sa}(\cH)$ can be characterized by 
    \[\|\gamma\|_1 = \inf \{\Tr\big(\gamma^+) + \Tr\big(\gamma^-) | \gamma^\pm \in \cS_1^+(\cH), \gamma= \gamma^+ - \gamma^-\}\] for all $\gamma\in \cS_1^\mathrm{sa}(\cH)$.
\end{lemma}
\begin{proof}
    By the triangle inequality of the trace norm, we have
    \[\|\gamma\|_1 =\|\gamma^+-\gamma^-\|_1\leq \|\gamma^+\|_1 + \|\gamma^-\|_1 =\Tr\big(\gamma^+)+\Tr\big(\gamma^-)\] for any decomposition $\gamma = \gamma^+ - \gamma^-$ satisfying $\gamma^\pm\in \cS_1^+(\cH)$. In particular, this infimum is attained by the choice $\gamma^\pm = (|\gamma|\pm\gamma)/2$, which gives $\Tr(\gamma^+)+\Tr(\gamma^-) = \Tr(|\gamma|)=\|\gamma\|_1.$ 
\end{proof}

We introduce the following notions of convergence for a sequence of trace class operators $\{A_n\}_{n\in\N}$:
 \begin{itemize}
     \item $A_n$ converges in the $\cS_1(\cH)$-norm to $A$, denoted $A_n \to A$, if \[\lim_{n\to\infty} \|A_n-A\|_1 =0.\]
     \item $A_n$ converges weak topology of $\cS_1(\cH)$ to $A$, denoted $A_n \rightharpoonup A$, if \[\lim_{n\to\infty} \Tr(V(A_n-A)) =0 \text{ for all } V\in \cB(\cH).\]
     \item $A_n$ converges in the weak-$*$ topology of $\cS_1(\cH)$ to $A$, denoted $A_n \xrightharpoonup{*} A$, if \[\lim_{n\to\infty} \Tr(V(A_n-A)) =0 \text{ for all }V\in \cK(\cH).\]
 \end{itemize}
Note that these notions of convergence satisfy \[(A_n\to A)\Rightarrow (A_n\rightharpoonup A)\Rightarrow (A_n\xrightharpoonup{*} A).\] As a consequence, the corresponding notions of continuity satisfy the opposite implication, i.e. for a function $f$ taking as argument operators, we have \[(\text{$f$ is weak-$*$ continuous})\Rightarrow (\text{$f$ is weakly continuous})\Rightarrow (\text{$f$ is norm-continuous}).\] Since $\cS_1(\cH)$ is not a reflexive Banach space, weak-$*$ convergence is strictly weaker than weak convergence. 

Next, define $\Tr_{p}^N: \cS_1\bigl( \bigotimes^N\cH\bigr) \to \cS_1(\bigotimes^{N-p}\cH)$ on pure tensor products by
\[
    \Tr_{p}^N (A_1 \otimes \cdots \otimes A_N) :=  \Tr(A_1) \cdots \Tr(A_p)  A_{p+1}\otimes\cdots \otimes A_N,
\]
where $A_1,\cdots,A_N\in \cS_1(\cH)$, and extend 
$\Tr_{p}^N$
linearly to all $\cS_1\bigl(\bigotimes^N\cH\bigr)$. 
We view $\cH^N$ as a closed linear subspace of $\bigotimes^N\cH$ and regard $\Tr_{p}^N $ as the restricted map  $\cS_1(\cH^N)\to \cS_1(\cH^{N-p})$. We call $\Pi:=N \Tr_{N-1}^N$ the \emph{partial trace map}. This map is in fact the adjoint of the second quantization map with respect to the dual pairing of $\cS_1(\cH)$ and $\cB(\cH)$ given by the trace: 
\begin{lemma}\label{Lemma: second quantization expectation value}
    If $V\in \cB(\cH)$ and $A\in \cS_1\left(\bigwedge^N\cH\right)$, then 
    \[\Tr(V \Pi(A))=\Tr\big(\hat V A).\]
\end{lemma}
\begin{proof}
See \cite[Theorem 10]{Kummer1967}.    
\end{proof}

\subsection{Sesquilinear forms and their associated operators} \label{Subsection: Sesquilinear forms and their associated operators}
In the following, we introduce some basic facts on sesquilinear forms and their associated operators. Comprehensive sources for this can be found in e.g.~\cite{Schmudgen2012} or~\cite{reed_methods_1980}.  A \emph{sesquilinear form} on $\cH$ is a map $v: D(v)\times D(v) \to \C$ which is antilinear in the first argument and linear in the second, and the the domain $D(v)$ is a dense linear subspace of $\cH$. We call $v$ \emph{Hermitian}\footnote{Also known as ``symmetric''.} if $v(\psi,\varphi)=\overline{v(\varphi,\psi)}$ for all $\psi,\varphi\in D(v)$. We say that $v$ is \emph{lower semibounded} if there exists a constant $\mu_v>0$ such that $\Re v(\psi,\psi)+\mu_v\|\psi\|^2>0$ for all $\psi\in D(v)$. Furthermore, we say that $v$ is \emph{closed} if $v$ is Hermitian and if $D(v)$ is complete with respect to the norm $\|\cdot\|_v$ induced by the inner product
\[
\langle \psi,\varphi\rangle_v = v(\psi,\varphi)+(\mu_v+1)\langle \psi,\varphi\rangle .
\]
If $v$ is closed, then it is bounded with respect to $\|\cdot\|_v$ in the sense that
\[
|v(\psi,\varphi)|\leq \|\psi\|_v\|\varphi\|_v
\]
for all $\psi,\varphi\in D(v)$.

Let $V : D(V) \to \cH$ be a lower semibounded self-adjoint operator with lower bound $\mu_V$. The \emph{form domain} of $V$ is defined by $Q(V) := D((V + \mu_V I)^{1/2})$. We denote the \emph{sesquilinear form associated to $V$} by a small letter $v$. It has domain $D(v)=Q(V)$, and is defined by
\[
v(\psi,\varphi)=\int\lambda \langle \psi,dP_V(\lambda)\varphi\rangle,
\]
where $dP_V$ is the spectral measure of $V$. If $\psi,\varphi\in D(V)$, we have $
v(\psi,\varphi) = \braket{\psi,V\varphi}. $

The sesquilinear form associated with a lower semibounded self-adjoint operator is lower semibounded, Hermitian and closed. A converse result \cite[Theorem 19.7]{Schmudgen2012}, sometimes called Kato's First Representation Theorem, states that given a lower semibounded Hermitian and closed form $v$, there is a unique lower semibounded self-adjoint operator $V$ such that $v$ is the sesquilinear form associated to $V$. There is thus a one-to-one correspondence between lower semibounded self-adjoint operators and lower semibounded, Hermitian and closed sesquilinear forms.

Let $t$ and $v$ be Hermitian sesquilinear forms. We say that $v$ is relatively $t$-bounded if $D(t)\subset D(v)$ and if there exist $a\geq 0$ and $b\geq 0$ such that
\[|v(\psi,\psi)|\leq at(\psi, \psi)+b\|\psi\|^2 \]
for all $\psi \in D(t)$. We call $a$ a \emph{$t$-bound} of $v$. This allows us to formulate the useful Kato--Lions--Lax--Milgram--Nelson (KLMN) theorem \cite[Theorem 10.21]{Schmudgen2012}: 
\begin{theorem}[KLMN theorem] Let $t$ be a closed lower semibounded Hermitian sesquilinear form and $T$ its associated operator. Let $v$ be a relatively $t$-bounded Hermitian sesquilinear form with $t$-bound $a<1$. Then, the sesquilinear form $c$ with domain $D(c)=D(t)$ defined by \begin{equation}\label{Equation: KLMN} c(\psi,\varphi)=t(\psi,\varphi) + v(\psi,\varphi)\end{equation}
    for all $\psi,\varphi\in D(t)$ is closed. The corresponding self-adjoint operator $C$ is called the \emph{form sum} of $t$ and $v$, and we use the notation $C=T\dot+v$. $C$ has lower bound $-b$. \label{Theorem: KLMN}
\end{theorem}
{\bfseries Warning:} The KLMN theorem does not state anything about the domain $D(T\dot{+} v)$, which may be different from $D(T)$. Furthermore, $v$ may itself not be associated to any operator, since it is not necessarily closed or closable. 

\medskip

If $T$ is positive, we introduce a normed space $\cF(T)$ of all relatively $T$-bounded Hermitian forms on $Q(T)$ by, i.e.,
\begin{equation}
\cF(T):= \{v: Q(T)\times Q(T)\to \C \text{ Hermitian}\,|\, \exists C_v>0 : \forall \psi\in Q(T): |v(\psi,\psi)|\leq C_v\|\psi\|_{T}^2\},
\label{eq:set of relatively bounded forms}
\end{equation}
where $\|\psi\|_{T}^2 := \|\psi\|^2+\|T^{1/2}\psi\|^2.$
The norm on $\cF(T)$ is defined 
\begin{equation}
\|v\|_{\cF(T)} := \sup_{\psi\in Q(T)}\frac{|v(\psi,\psi)|}{\|\psi\|_{T}^2}. \label{Equation: norm on cF(T)}
\end{equation}


\subsection{Convex analysis} \label{Subsection: Convex analysis}

See, e.g., Ref.~\cite{borwein_convex_2010} for more on the following. 
For this section, let $E$ be a topological vector space. Let $S\subset E$. The \emph{affine hull} $\Aff(S)$ is the smallest affine space containing $S$, and consists of finite linear combinations $\sum_{i=1}^n \lambda_i s_i$ of elements $s_i \in S$ such that $\sum_{i=1}^n\lambda_i = 1$. The \emph{relative interior} $\operatorname{ri}(S)$ of $S$ is the interior of $S$ in $\Aff(S)$ in the subspace topology. The \emph{convex hull} $\Conv(S)$ is defined as the smallest convex subset of $E$ contianing $S$, which consists of the set of finite linear combinations $\sum_{i=1}^n \lambda_i s_i$ of elements $s_i \in S$ such that $\sum_{i=1}^n\lambda_i = 1$ with the additional condition that $\lambda_i\geq 0$. 

A function $f : E\to \R$ is said to be \emph{lower semicontinuous} if whenever $x_n\to x$, we have $f(x)\leq \liminf_{n\to \infty}f(x_n)$. A function $f: E\to \R$ is said to be convex if for any $x_1,x_2\in E$ and $0\leq \alpha\leq 1$, we have
\[f(\alpha x_1 + (1-\alpha)x_2) \leq \alpha f(x_1) + (1-\alpha)f(x_2).\]

The convex envelope $\CE (g)$ of a function $g: E\to \R$ is defined  the largest convex lower semicontinuous function satisfying $\CE(g)\leq g$.

\section{Reduced Density Matrices}\label{Section: Reduced Density Matrices}
In this section, we introduce a real Banach space $X\subset \cS_1^\text{sa}(\cH)$ containing the reduced density matrices $\gamma$ with finite kinetic energy. The space $X$ is chosen such that the dual $X^*$ contains certain unbounded operators, which we identify as the admissible potentials of our theory. We take $T$ to denote the kinetic energy operator, which for us is an arbitrary positive self-adjoint operator on the one-particle space $\cH$. The reader may think of the electronic-structure theory setting, where $\cH = L^2(\R^3\times\Z_2)$ and $T = -\nabla^2$. Since $T$ is unbounded, we take particular care in defining expressions such as $\Tr(T\gamma)$, which we do in the following. For $\gamma\in \cS_1^+(\cH)$, we define 
\[\Tr(T\gamma):= \sup_{\epsilon>0} \Tr(T_\epsilon \gamma ),\]
where $T_\epsilon := T(1+\epsilon T)^{-1}$ is defined using the spectral calculus. Since $\{T_\epsilon\}_{\epsilon>0}$ is a monotonically increasing family of bounded operators, we have $\sup_{\epsilon>0} \Tr(T_\epsilon \gamma) = \lim_{\epsilon \to 0^+}\Tr(T_\epsilon \gamma).$ It follows from this that if $\gamma,\gamma'\in \cS_1^+(\cH)$, then we have
\begin{equation} \label{Equation: Split the T}
\Tr(T(\gamma+\gamma'))=\Tr(T\gamma)+\Tr(T\gamma').
\end{equation}
Since the function $x\mapsto x/(1+\epsilon x)$ is operator monotone for any $\epsilon >0$ \cite[Lemma 4.20, Chapter II]{Takesaki2003}, it follows that $T'\leq T \Rightarrow\Tr(T'\gamma)\leq \Tr(T\gamma)$ for all $\gamma\in \cS_1^+(\cH)$. This leads to the following useful characterization: 
\begin{lemma}
Let $V\geq 0$ be a possibly unbounded operator, and let $\gamma\in \cS_1^+(\cH)$. Then, we have
    \begin{equation} \label{Equation: T-trace is supremum over compacts}
    \Tr(V\gamma ) = \sup \left\{\Tr(V'\gamma) \;\middle|\; V' \in \cK^+(\cH), \; V' < V \right\}.
\end{equation}
\end{lemma}
\begin{proof}
     One can show that since $V'<V$, we have $\sup\{\Tr(V'\gamma)\,|\,0\leq V'<V,V'\in \cK(\cH)\}\leq \Tr(V\gamma)$. By definition, $\Tr(V\gamma)=\sup_{\epsilon>0}\Tr(V_\epsilon\gamma)$, and it is a basic property of the trace that $\Tr(\gamma)= \sup_P \Tr(P\gamma)$, where the supremum is taken over all finite-rank projections $P$. For any such projection, $V_\epsilon^{1/2} P V_\epsilon^{1/2}$ is finite-rank and hence compact. Thus, 
    \[\Tr(V\gamma)=\sup_{\epsilon>0}\sup_{P}\Tr(V_\epsilon^{1/2} P V_\epsilon^{1/2} \gamma).\]
  Since $P\leq 1,$ we have $V_\epsilon^{1/2} P V_\epsilon^{1/2}\leq V_\epsilon < V$, and hence \[\sup_{\epsilon>0}\sup_{P}\Tr(V_\epsilon^{1/2} P V_\epsilon^{1/2} \gamma)\leq \sup\bigl\{\Tr(V'\gamma)\;\big|\; V' \in \cK^+(\cH), \;  V' <V\bigr\}, \] which proves the lemma. 
\end{proof}
The following technical lemma, which can be regarded as a version of Fatou's Lemma will be useful later: 
\begin{lemma}\label{Lemma: Trace of unbounded operator is lower semicontinuous}
    Let $V\geq 0$ be a possibly unbounded operator. If $\{\gamma_n\}_{n\in\N}$ is a sequence in $\cS_1^+(\cH)$ which converges in the weak-$*$ topology of $\cS_1(\cH)$ to some $\gamma\in \cS_1^+(\cH)$, then
    \[\Tr(V\gamma)\leq \liminf_{n\to\infty}\Tr(V\gamma_n).\]
\end{lemma}
\begin{proof}
    For every $V'\in \cK^+(\cH)$ satisfying $V'< V$, we have
    \[\Tr(V'\gamma_n)\leq\Tr(V\gamma_n)\] for every $n$. Taking the limit inferior on both sides, we get
    \[\Tr(V'\gamma ) \leq \liminf_{n\to \infty} \Tr(V\gamma_n).\] Taking the supremum over $V'$, we get the result by \eqref{Equation: T-trace is supremum over compacts}.
\end{proof}

We are now ready to define our space of finite-kinetic energy density matrices.
Let
\[ X^+ := \{ \gamma \in \cS_1^+(\cH) \mid \Tr(T\gamma) < +\infty \}, \]
and define
\[X := X^+ - X^+ = \{\gamma^+-\gamma^- \mid \gamma^\pm\in \cS_1^+(\cH), \Tr(T\gamma^\pm)<+\infty\}.\]
Intuitively, $X^+$ is the convex cone of density matrices with finite kinetic energy, and $X$ is the set of all trace-class operators that can be decomposed into a positive and negative part, each with finite kinetic energy. This space will be turned into a real Banach space in the following.

\begin{lemma}\label{Lemma: X is dense in S}
    $X$ is dense in $\cS^\text{sa}_1(\cH)$.
\end{lemma}
\begin{proof} In this proof, we use the following fact: For any $\varphi,\psi\in \cH$, which satisfy $\|\varphi\|=\|\psi\|=1$, we have
    \begin{equation}\label{Equation: Super basic trace inequality}
        \Tr(\big||\varphi\rangle\langle\varphi|-|\psi\rangle\langle \psi|\big|)\leq 2\|\varphi-\psi\|.
    \end{equation}
This elementary inequality will be demonstrated in Appendix \ref{Appendix: Basic Trace Inequality}. Let $0\neq \gamma\in \cS_1(\cH)^\mathrm{sa}$, and denote by 
    \[\gamma = \sum_n\lambda_n|\varphi_n\rangle\langle \varphi_n|\] 
its unitary diagonalization. Next, recall that since $T$ is positive, the form domain of $T$ equals $Q(T)=D(T^{1/2})$. If $\varphi\in \cH$, then $\varphi\in Q(T)$ iff $\Tr(T|\varphi\rangle\langle \varphi|)<+\infty$, and if $\varphi\in Q(T)$, then $\Tr(T|\varphi\rangle\langle \varphi|)=\|T^{1/2}\varphi\|^2$. 
    
    Keeping the above facts in mind, we proceed with proving the lemma. Let $\epsilon>0$. We will show that there is a $\gamma_\epsilon\in X$ such that $\|\gamma-\gamma_\epsilon\|_1<\epsilon$. Since $D(T^{1/2})$ is dense in $\cH$, we can find for every $n$ a $\varphi_n^\epsilon\in D(T^{1/2})$ such that $\|\varphi_n-\varphi_n^\epsilon\|<\epsilon/(4\Tr|\gamma|)$ and $\|\varphi_n^\epsilon\|=1$. 
    For large enough $N_\epsilon\in  \N$, we have
    \[\sum_{n=N_\epsilon+1}^\infty|\lambda_n|<\epsilon/2.\]
    Denote 
    \[\gamma_{\epsilon} := \sum_{n=1}^{N_\epsilon} \lambda_n|\varphi_n^\epsilon\rangle \langle \varphi_n^\epsilon|.\] Then $\gamma_\epsilon\in X$ since $\Tr(T|\varphi^\epsilon_n\rangle \langle \varphi_n^\epsilon|)=\|T^{1/2}\varphi_n^\epsilon\|^2<+\infty $ for all $n$. Using \eqref{Equation: Super basic trace inequality}, we have
    \[\Tr(\big||\varphi_n\rangle\langle \varphi_n|-|\varphi_n^\epsilon\rangle\langle \varphi_n^\epsilon|\big|)\leq 2\|\varphi_n-\varphi_n^\epsilon\|<\frac{\epsilon}{2\Tr(|\gamma|)}\]
    for every $n$. Combining this with the triangle inequality, we get
    \[\|\gamma-\gamma_\epsilon\|_1 =\Tr(|\gamma-\gamma_\epsilon|)\leq \sum_{n =1}^{N_\epsilon}|\lambda_n| \Tr(\big||\varphi_n\rangle\langle \varphi_n|-|\varphi_n^\epsilon\rangle\langle \varphi_n^\epsilon|\big|) + \sum_{n=N_\epsilon+1}^\infty|\lambda_n|\]
    \[< \sum_{n =1}^{N_\epsilon}|\lambda_n|\frac{\epsilon}{2\Tr(|\gamma|)} + \frac{\epsilon}{2}\leq \epsilon,\] which was what we wanted to show.
\end{proof}
Define the map $\|\cdot\|_X:\cS_1(\cH)\to \R\cup\{+\infty\}$ by 
\[\|\gamma\|_X :=  \inf \left\{ \Tr\big((1+T)\gamma^+\big)+\Tr\big((1+T)\gamma^-\big) \;\middle|\; \gamma^\pm\in \cS_1^+(\cH), \gamma = \gamma^+ - \gamma^-\right\}.\]
The fact that $\|\cdot\|_X$ is a norm when restricted to $X$ follows from general results in operator algebra. See for example \cite[Ch. VII, \S1, Lemma 1.9]{Takesaki2003} for a more general construction, or \cite{Sherstnev2012} and \cite{Stolyarov2002}. We provide an elementary proof here for completeness:
\begin{lemma}\label{Lemma: The X-norm}
   $\|\cdot\|_X$ defines a norm on $X$ which satisfies $\|\gamma\|_1\leq \|\gamma\|_X$ for all $\gamma\in X$.  
\end{lemma}
\begin{proof} The fact that $\|\cdot\|_1 \leq \|\cdot\|_X$ follows from Lemma \ref{lemma: gauge characterization of trace norm} and the fact that $T$ is positive. It follows that if $\|\gamma\|_X=0$, then $\|\gamma\|_1=0$, and hence $\gamma= 0$. 

By definition of $X$, it is clear that $\|\gamma\|_X<+\infty$ for all $\gamma\in X$. Furthermore, since $T$ is positive, we have $\|\lambda\gamma\|_X =  \lambda\|\gamma\|_X$ if $\lambda>0$, and if $\lambda<0$, we have 
\begin{equation*}
\begin{split}
    \|\lambda\gamma\|_X &= \inf \left\{ \Tr\big((1+T)\gamma^+\big)+\Tr\big((1+T)\gamma^-\big) \;\middle|\; \gamma^\pm\in \cS_1^+(\cH), \lambda\gamma = \gamma^+ - \gamma^-\right\} \\
&= \inf \left\{ \Tr\big((1+T)\gamma^+\big)+\Tr\big((1+T)\gamma^-\big) \;\middle|\; \gamma^\pm\in \cS_1^+(\cH), \gamma = (- \lambda^{-1}) \gamma^- - (-\lambda^{-1})\gamma^+\right\} \\
& = \inf \left\{ -\lambda \Tr\big((1+T)\tilde\gamma^+\big)-\lambda \Tr\big((1+T)\tilde\gamma^-\big) \;\middle|\; \tilde \gamma^\pm\in \cS_1^+(\cH), \gamma = \tilde\gamma^+-\tilde\gamma^-\right\}
\\ &= -\lambda\|\gamma\|_X,
\end{split}
\end{equation*}
so $\|\lambda\gamma\|_X = |\lambda|\|\gamma\|_X$ for all $\lambda\in\R$ and $\gamma\in X$. It remains to demonstrate the triangle inequality. Let $\gamma,\tilde\gamma\in X$, and let $\gamma = \gamma^+-\gamma^-$ and $\tilde \gamma= \tilde\gamma^+-\tilde\gamma^-$ be decompositions such that $\gamma^\pm \geq 0$, $\tilde\gamma^\pm\geq0$, $\Tr(T\gamma^\pm)<+\infty$ and $\Tr(T\tilde{\gamma}^\pm)<+\infty$. We then have that $\gamma+\tilde\gamma = \gamma^++\tilde\gamma^+-(\gamma^-+\tilde\gamma^-)$. Hence
\begin{equation*}
    \begin{split}
        \|\gamma+\tilde\gamma\|_X &= \inf\left\{\Tr((1+T)\hat\gamma^+)+\Tr((1+T)\gamma^-)\;\middle|\; \hat{\gamma}^{\pm} \in \cS_1^+(\cH), \; \gamma+\gamma'= \hat\gamma^+-\hat\gamma^-\right\} \\
        & \leq \inf \Bigl\{\Tr((1+T)(\gamma^++\tilde\gamma^+))+\Tr((1+T)(\gamma^-+\tilde\gamma^-)) \;\Big|\; \\& \qquad\qquad\qquad\qquad\qquad\qquad\qquad \gamma^{\pm},\tilde{\gamma}^{\pm} \in \cS_1^+(\cH), \; \gamma = \gamma^+-\gamma^-,\tilde\gamma= \tilde\gamma^+-\tilde\gamma^-\Bigr\}
        \\ &=\|\gamma\|_X+\|\tilde\gamma\|_X.
    \end{split}
\end{equation*}
The inequality above comes from the fact that the infimum is taken over a smaller set.  
\end{proof}

\begin{lemma}\label{Corollary: Unique decomposition}
    Let $\gamma\in X$. There exists a decomposition $\gamma = \gamma^+_0-\gamma^-_0$ such that $\|\gamma\|_X = \Tr((1+T)\gamma^+_0)+\Tr((1+T)\gamma^-_0)$. 
\end{lemma}
\begin{proof}
For every $n$, pick decompositions $\gamma = \gamma^+_n - \gamma^-_n$ satisfying 
\[ \Tr((1+T)\gamma_n^+) +\Tr((1+T)\gamma_n^-)\leq \|\gamma\|_X +\frac{1}{n}. \]
The sequences $\{\gamma^\pm_n\}_{n\in\N}$ are bounded in the trace norm. Picking a subsequence if necessary, there exist weak-$*$ limits $\gamma^\pm_0\in \cS_1^\mathrm{sa}(\cH)$ by the Banach--Alaoglu Theorem. Notice that we have for any compact $K\in \cK(\cH)$ that 
\[\Tr(K\gamma)=\lim_{n\to\infty}\Tr(K(\gamma^+_n-\gamma^-_n))=\Tr(K(\gamma^+_0-\gamma^-_0)), \]
whence it follows that $\gamma= \gamma^+_0 -\gamma^-_0$. For any compact $(1+T)'<(1+T)$ we have 
\[ \Tr((1+T)'\gamma_n^+) +\Tr((1+T)'\gamma_n^-) \leq \Tr((1+T)\gamma_n^+)+\Tr((1+T)\gamma_n^-)\leq \|\gamma\|_X +\frac{1}{n}.\]
Taking the limit in $n$, and supremum over all compacts, we get by Lemma \ref{Equation: T-trace is supremum over compacts} that 
\[\Tr((1+T)\gamma^+_0) + \Tr((1+T)\gamma^-_0)\leq \|\gamma\|_X.\] Since $\gamma= \gamma^+ -\gamma^-$, it follows that \[\Tr((1+T)\gamma^+_0) + \Tr((1+T)\gamma^-_0) = \|\gamma\|_X,\] which is what we wanted to show. 
\end{proof}
We can thus replace the ``inf'' with a ``min'' in the definition of $\|\cdot\|_X$. 
\begin{remark} The above lemma says nothing about the uniqueness of the decomposition $\gamma = \gamma^+ - \gamma^-$. We note that by the choice $T=0$, Lemma \ref{Corollary: Unique decomposition} gives the unique Jordan decomposition $\gamma^\pm = (|\gamma|\pm\gamma)/2$ if we also require $\gamma^\pm\gamma^\mp=0$. Uniqueness might not be guaranteed in general. 
\end{remark}

\begin{remark}
    The task of evaluating $\|\gamma\|_X$ for a given $\gamma\in \cS_1^\mathrm{sa}(\cH)$ can be formulated as a semidefinite program. Indeed, by a change of variables, we obtain the semidefinite program 
    \[\|\gamma\|_X = \inf \{\Tr((1+T)Z) \;|\; Z\in \cS_1^\mathrm{sa}(\cH),\; Z-\gamma\geq 0,\;Z+\gamma\geq 0\},\] which is numerically solvable with standard methods in the finite-dimensional case \cite{Boyd_Vandenberghe_2004}. 
\end{remark}

The functional $\gamma \mapsto \Tr(T\gamma)$ is \emph{a priori} only defined for positive $\gamma$. We expand it to non-positive $\gamma\in X$ by defining 
\begin{equation}\label{Equation: Tr(T gamma)}    \Tr\big(T\gamma\big):=\Tr\big(T\gamma^+\big)-\Tr\big(T\gamma^-\big),
\end{equation}
where $\gamma = \gamma^+ - \gamma^-$, is any decomposition of $\gamma$ which satisfies $\gamma^\pm \in X^+$.  
\begin{lemma}\label{Lemma: Trace inequality}  The map $\Tr(T\,\cdot\,):X\to \R$ as given by \eqref{Equation: Tr(T gamma)} is well-defined, linear and bounded. In particular, it is independent of the choice of decomposition. Furthermore, we have
    \[\Tr(|\gamma|) + |\Tr(T\gamma)|\leq \|\gamma\|_X\leq \Tr(|\gamma|) + \Tr(T|\gamma|)\] 
for any $\gamma\in X$. Both inequalities are equalities when $\gamma$ is positive. 
\end{lemma}
\begin{proof} We first demonstrate the fact that $\Tr(T\,\cdot\,)$ is independent of the choice of decomposition. Let $\gamma = \gamma^+ -\gamma^-= \tilde\gamma^+-\tilde \gamma^-$ be any two decompositions of $\gamma$ satisfying $\gamma^\pm\geq 0$ and $\Tr(T\gamma^\pm) <+\infty$. Then, we have 
\[ \gamma^+ + \tilde\gamma^-= \tilde\gamma^+ + \gamma^-.\] 
This implies that 
\[\Tr\big(T(\gamma^+ + \tilde\gamma^-)\big) = \Tr\big(T(\tilde\gamma^+ + \gamma^-)\big),\]
and by \eqref{Equation: Split the T} we thus get
\[\Tr\big(T\gamma^+)-\Tr\big(T\gamma^-\big)=\Tr\big(T\tilde\gamma^+\big)-\Tr\big(T\tilde\gamma^-\big),\] showing that $\Tr(T\cdot)$ is indeed well-defined and independent of the choice of decomposition of $\gamma$. The linearity of $\Tr(T\cdot)$ follows from this fact and \eqref{Equation: Split the T}. 

In order to demonstrate the first inequality of the lemma, observe that for any decomposition $\gamma = \gamma^+ -\gamma^-$ with $\Tr(T\gamma^\pm)<+\infty$, we have 
\[|\Tr\big(T\gamma\big)|= |\Tr\big(T\gamma^+\big)-\Tr\big(T\gamma^-\big)|\leq \Tr\big(T\gamma^+\big)+\Tr\big(T\gamma^-\big).\]
By the triangle inequality for the trace norm, we have that 
\[\Tr\big(|\gamma|\big)\leq \Tr\big(\gamma^+\big) + \Tr\big(\gamma^-\big),\]
such that 
\[\Tr(|\gamma|)+|\Tr(T\gamma)|\leq \Tr\big((1+T)\gamma^+\big) + \Tr\big((1+T)\gamma^-\big).\]
Taking the infimum over all such decompositions yields the first inequality. Next, observe that if $\gamma\in X$, and if $\Tr(T|\gamma|)$ is finite, it follows that $\gamma^\pm = (|\gamma|\pm \gamma)/2$ is a valid decomposition of $\gamma$, and satisfies $\Tr((1+T)\gamma^+) + \Tr((1+T)\gamma^- )= \Tr(|\gamma|) + \Tr(T|\gamma|).$ Hence, 
\begin{equation*}
    \begin{split}
    \|\gamma\|_X &=\inf \bigl\{ \Tr\big((1+T)\gamma^+\big)+\Tr\big((1+T)\gamma^-\big) \;\big|\; \gamma = \gamma^+ - \gamma^-, \gamma^\pm \geq 0\bigr\} \\&\leq \Tr(|\gamma|)+\Tr(T|\gamma|),   
    \end{split}
\end{equation*}
which gives the second inequality.
\end{proof}
\begin{remark}
    The inequalities in Lemma \ref{Lemma: Trace inequality} do not hold if we consider $\gamma$ to be non-self-adjoint. Indeed, if the above inequalities were to hold for all non-self-adjoint $\gamma$'s, then $|\Tr(T\gamma)|\leq \Tr(T|\gamma|)$ would imply $T=\lambda \one$ for some $\lambda \in \R^+$ \cite[Theorem 4]{Stolyarov2002}. \end{remark}
\begin{example}
    We illustrate the above remark with an example of $\|\gamma\|_X< \Tr(|\gamma|) +\Tr(T|\gamma|)$. Indeed, let 
    \[T=\begin{pmatrix}0 & 0 \\ 0 & 3 \end{pmatrix}\text{ and }\gamma=\begin{pmatrix}0 & 1 \\ 1 & 0 \end{pmatrix}.\]
    Then, $\Tr(|\gamma|) + \Tr(T|\gamma|)=2+3=5$. Pick the decomposition $\gamma=\gamma^+-\gamma^-$, where
\[\gamma^+=\begin{pmatrix}1 & 0.5 \\ 0.5 & 0.25 \end{pmatrix}\text{ and }\gamma^-=\begin{pmatrix}1 & -0.5 \\ -0.5 & 0.25 \end{pmatrix}.\]
    One readily verifies that $\gamma^\pm\geq 0$. Using this decomposition, we get 
    \[\Tr((1+T)\gamma^+)+\Tr((1+T)\gamma^-)=2+2 = 4 ,\] which is strictly smaller than $\Tr(|\gamma|)+\Tr(T|\gamma|)=5.$ Hence $\|\gamma\|_X\leq 4<5=\Tr(|\gamma|)+\Tr(T|\gamma|)$. In fact, for this example, it can be shown that this decomposition is the unique minimizer of $\|\gamma\|_X$.
\end{example}
\begin{example}
Lemma \ref{Lemma: Trace inequality} ensures that $\gamma\in X$ if $\Tr((1+T)|\gamma|)<\infty$. However, the opposite does not hold. We illustrate this with an example for which $\|\gamma\|_X<\infty$ while $\Tr((1+T)|\gamma|)=\infty$. Let $\cH = \bigoplus_{n=1}^\infty\C^2$, and define
\[T=\bigoplus_{n=1}^\infty T_n, \text{ where } T_n=\begin{pmatrix} 0 & 0 \\ 0 & n^4 \end{pmatrix} \text{ and }\gamma=\bigoplus_{n=1}^\infty\gamma_n, \text{ where } \gamma_n = \begin{pmatrix} 0 & n^{-4} \\ n^{-4} & 0 \end{pmatrix}.\]
Since 
\[|\gamma_n| = \begin{pmatrix} n^{-4} & 0 \\ 0 & n^{-4}\end{pmatrix},\] we have
\[\Tr(T_n|\gamma_n|)=\Tr\left(\begin{pmatrix} 0 & 0 \\ 0 & n^4 \end{pmatrix}\begin{pmatrix} n^{-4} & 0 \\ 0 & n^{-4}\end{pmatrix}\right)= 1,\]
which implies that 
\[\Tr(T|\gamma|)=\sum_{n=1}^\infty\Tr(T_n|\gamma_n|)=\sum_{n=1}^\infty 1 = \infty. \] Conversely, we can check that $\|\gamma\|_X$ is finite. Indeed, define $\gamma_n^\pm = \pm \gamma_n/2 + \Delta_n$, where
\[\Delta_n = \begin{pmatrix} \frac{1}{4n^2} & 0 \\ 0 & \frac{1}{n^6} \end{pmatrix}.\] Then, it is straightforward to check that $\gamma_n = \gamma_n^+-\gamma_n^-$ and $\gamma_n^\pm \geq 0$. Defining $\gamma^\pm = \bigoplus \gamma_n^\pm$, we have $\gamma = \gamma^+-\gamma^-$ and $\gamma^\pm\geq0$, and so 
\[\|\gamma\|_X\leq \Tr((1+T)(\gamma^++\gamma^-))=\sum_{n=1}^\infty \Tr((1+T_n)2\Delta_n)= \sum_{n=1}^\infty \left(\frac{5}{2}n^{-2}+2n^{-6}\right) <\infty.\]
\end{example}
\begin{lemma}\label{Lemma: the diagonalization is in D(T^(1/2))}
    Let $\gamma\in X$. Then, the diagonalization given by the Spectral Theorem, 
\[\gamma = \sum_n\lambda_n |\varphi_n\rangle\langle\varphi_n|,\]
satisfies $\varphi_n\in Q(T)$ for all $n$ for which $\lambda_n\neq0$. 
\end{lemma}
\begin{proof} First, suppose that $\gamma\in X^+$. Assuming $\lambda_n\neq  0$ for all $n$, we have 
\[+\infty>\|\gamma\|_X=\Tr((1+T)\gamma)=\sum_n\lambda_n \Tr((1+T)|\varphi_n\rangle\langle \varphi_n|) = \sum_n \lambda_n \|\varphi_n\|^2_T,\]
implying that $\varphi_n\in Q(T)$ for all $n$. Next, we show that the range of $\gamma$ is contained in $Q(T)$. Indeed, let $\psi\in \cH$. Then, using Hölder's inequality and Parseval's formula, together with the fact that $\lambda_n\leq \|\gamma\|_1$ for all $n$, we have
\begin{align*}
    \|\gamma\psi\|_T^2 &= \left\|\sum_{n=1}^\infty \lambda_n \langle \varphi_n,\psi\rangle \varphi_n\right\|_T^2\leq \left(\sum_{n=1}^\infty\lambda_n|\langle \varphi_n,\psi\rangle |\|\varphi_n\|_T\right)^2 \\ &\leq\sum_{n=1}^\infty|\langle \varphi_n,\psi\rangle |^2 \sum_{n=1}^\infty\lambda_n^2 \|\varphi_n\|_T^2 \leq \|\psi\|^2 \|\gamma\|_1\sum_{n=1}^\infty\lambda_n \|\varphi_n\|_T^2 \\ &=\|\psi\|^2\|\gamma\|_1\|\gamma\|_X<\infty,
\end{align*}
implying that $\ran \gamma\subset Q(T)$. Next, suppose that $\gamma\in X$, and let $\gamma= \gamma^+-\gamma^-$ be any decomposition with $\gamma^\pm\in X^+$. Then if $\psi\in \cH$ is an eigenvector of $\gamma$ with eigenvalue $\lambda$, we have
\[\gamma \psi=\lambda\psi = \gamma^+\psi-\gamma^-\psi,\] implying that $\psi\in Q(T)$ since $\ran\gamma^\pm\subset Q(T)$. 
\end{proof}
The above lemma has a natural corollary:
\begin{corollary}
    \label{Corollary: RD is dense in X} The set
    \[\R Q(T)\equiv \mathrm{span}_\R\left \{\;|\varphi\rangle\langle \varphi| \;\middle|\; \varphi\in Q(T) \right\}\]
is dense in $X$.
\end{corollary}

We will show in the following that $X$ is complete when equipped with the norm $\|\cdot\|_X$, and is hence a real Banach space. In order to prove this, we need the following theorem. 

\begin{theorem} \label{Theorem: The X-norm is lower semicontinuous with respect to the weak-$*$ topology}
    $\|\cdot\|_X : \cS_1(\cH) \to \R^+\cup\{+\infty\}$ is lower semi-continuous with respect to the weak-$*$ topology on $\cS_1(\cH)$. 
\end{theorem}

\begin{proof} Let $\{\gamma_n\}_{n\in \N}$ be a sequence in $\cS^\text{sa}_1(\cH)$ which converges in the weak-$*$ topology to some element $\gamma\in \cS_1^\text{sa}(\cH)$. We must show that 
\begin{equation}\label{Equation: Lower semicontinuity of T-semionorm}
    \|\gamma\|_X \leq \liminf_{n\to \infty} \|\gamma_n\|_X.
\end{equation} 
If $\{\gamma_n\}_{n\in\N}$ has no subsequence which is bounded in $\|\cdot\|_X$, then the right hand side of \eqref{Equation: Lower semicontinuity of T-semionorm} equals $+\infty$, and the theorem holds trivially. Thus, restrict $\{\gamma_n\}_{n\in\N}$ to a subsequence which is bounded in $\|\cdot\|_X$. By Lemma \ref{Corollary: Unique decomposition}, we can pick a decomposition $\gamma_n= \gamma_n^+-\gamma_n^-$ satisfying \begin{equation} \label{Equation: seminorm squeeze}
    \|\gamma_n\|_X  = \Tr\big((1+T)\gamma_n^+\big)+\Tr\big((1+T)\gamma_n^-\big)
\end{equation} for every $n\in\N$. Since $\gamma_n^+$ and $\gamma^-_n$ are positive, we have 
\[\|\gamma_n^\pm \|_1 = \Tr(\gamma_n^\pm)\leq \Tr((1+T)\gamma_n^\pm)\leq \|\gamma_n\|_X,\] which is bounded since $\|\gamma_n\|_X$ is bounded. Hence, by the Banach-Alaoglu theorem, there exist subsequences of $\{\gamma_n^\pm\}_{n\in\N}$ which converge in the weak-$*$ topology to some $\gamma^\pm\in \cS_1(\cH)$. We have
\[\Tr(A(\gamma^+-\gamma^-)) = \lim_{n\to \infty} \Tr(A(\gamma_n^+ -\gamma_n^-))=\Tr(A\gamma)\]
for all $A\in \cK(\cH)$, which means that $\gamma = \gamma^+ -\gamma^-$. Since $\gamma^\pm$ gives a decomposition of $\gamma$, we have by definition of $\|\cdot\|_X$ that 
\[\|\gamma\|_X \leq \Tr((1+T)\gamma^+)+\Tr((1+T)\gamma^-).\]
By Lemma \ref{Lemma: Trace of unbounded operator is lower semicontinuous}, we have
\begin{equation*}
\begin{split}
\Tr\big((1+T)\gamma^+\big)+\Tr\big((1+T)\gamma^-\big)&\leq \liminf_{n\to\infty}\Tr\big((1+T)\gamma^+_n\big)+\Tr\big((1+T)\gamma^-_n\big) \\ &= \liminf_{n\to\infty}\|\gamma_n\|_X,
\end{split}
\end{equation*}
the last equality being due to \eqref{Equation: seminorm squeeze}. 
\end{proof}

\begin{theorem}
    $(X,\|\cdot\|_X)$ is a real Banach space.
\end{theorem}
\begin{proof}
We need to show that $X$ is complete. Let $\{\gamma_n\}_{n\in\N}$ be a Cauchy sequence in $\|\cdot\|_X$. It is sufficient to show that $\{\gamma_n\}_{n\in N}$ has a convergent subsequence. Since $\{\gamma_n\}_{n\in\N}$ is Cauchy, it is in particular bounded. Hence, we can use the Banach-Alaoglu Theorem to restrict to a weak-$*$ convergent subsequence $\gamma_n\xrightharpoonup{*}\gamma$. By Theorem \ref{Theorem: The X-norm is lower semicontinuous with respect to the weak-$*$ topology}, we have
\[\|\gamma\|_X\leq \liminf_{n\to \infty}\|\gamma_{n}\|_X<+\infty,\] which means that $\gamma\in X$. 

Next, we prove that $\|\gamma_n-\gamma\|_X\to0$. Let $\epsilon>0$. Since $\{\gamma_n\}_{n\in \N}$ is Cauchy, we can find $N\in \N$ such that $\|\gamma_n-\gamma_{m}\|_X <\epsilon$ for all $m,n\geq N$. Since $\gamma_m-\gamma_n \xrightharpoonup[n\to\infty]{*}\gamma_m-\gamma$, it follows from Theorem \ref{Theorem: The X-norm is lower semicontinuous with respect to the weak-$*$ topology} that we have 
\[\|\gamma_m-\gamma\|_X\leq \liminf_{n\to \infty}\|\gamma_m-\gamma_n\|_X \leq\epsilon\] for all $m\geq N$. Since $\epsilon>0$ was arbitrary, this concludes the proof. 
\end{proof}

The Banach space $X$ will turn out to be a useful framework for density matrices. 

\section{Single-particle Hamiltonians}\label{Section: Single-particle Hamiltonians}

By definition, the dual space of $X$ is given by
\[X^* := \left\{ \omega: X\to \R  \;\middle | \; \omega \text{ is linear and } \exists C_\omega>0 : |\omega(\gamma)|\leq C_\omega \|\gamma\|_X  \text{ for all } \gamma\in X\right\}.\]

\begin{lemma}\label{Lemma: Unique linear extension}
    Let $\omega: X^+ \to \R$ be a map satisfying $|\omega(\gamma)|\leq C_\omega\|\gamma\|_X$, where 
    \[C_\omega := \sup_{\gamma\in X^+}\frac{|\omega(\gamma)|}{\|\gamma\|_X} < +\infty ,\]
   and $\omega(a\gamma + b\gamma')=a\omega(\gamma) + b\omega(\gamma')$ for all $\gamma,\gamma'\in X^+$ with $a,b\in \R^+$. Then, $\omega$ has a unique linear bounded extension to all of $X$, i.e. there is a unique $\overline{\omega}\in X^*$ satisfying $\overline{\omega}(\gamma)=\omega(\gamma)$ for all $\gamma\in X^+$. Furthermore, we have \begin{equation}\label{Equation: Dual norm from positive elements}
    \|\overline{\omega}\|_{X^*} = C_\omega.
\end{equation} 
\end{lemma}
\begin{proof}
    Let $\gamma\in X$, and let $\gamma= \gamma ^+ -\gamma^-$ be a decomposition such that $\gamma^\pm \in X^+$. We define 
    \[\overline{\omega}(\gamma):= \omega(\gamma^+)-\omega(\gamma^-).\]
    Firstly, this is independent of the choice of decomposition, since if $\gamma = \tilde\gamma^+-\tilde\gamma^-$ is any other decomposition, then 
    \[\gamma^+ -\gamma^- = \tilde\gamma^+ - \tilde\gamma^-\Leftrightarrow \gamma^++\tilde\gamma^-=\gamma^- +\tilde\gamma^+,\] and hence 
    \[\omega(\gamma^+)+\omega(\tilde\gamma^-)=\omega(\gamma^-) +\omega(\tilde\gamma^+) \Leftrightarrow \omega(\gamma^+)-\omega(\gamma^-)=\omega(\tilde\gamma^+)-\omega(\tilde\gamma^-).\]
    It is straightforward to check that $\overline{\omega}$ is linear and unique, and finally we have that
    \[|\overline{\omega}(\gamma)|\leq |\omega(\gamma^+)|+|\omega(\gamma^-)|\leq C_\omega\left(\Tr((1+T)\gamma^+)+\Tr((1+T)\gamma^-)\right).\] Taking the infimum over decompositions, we obtain that 
\[\|\overline{\omega}\|_{X^*}\leq C_\omega = \sup_{\gamma\in X^+}\frac{|\omega(\gamma)|}{\|\gamma\|_X}\leq \sup_{\gamma\in X}\frac{|\overline\omega(\gamma)|}{\|\gamma\|_X} = \|\overline{\omega}\|_{X^*},\]
    and hence, $\|\overline\omega\|_{X^*}=C_\omega$.
\end{proof}
We say that $\omega\in X^*$ has \emph{$T$-bound} $a_\omega>0$ if there exists $b_\omega>0$ such that
\begin{equation}\label{Equation: Relatively bounded} |\omega(\gamma)|\leq a_\omega\Tr(T\gamma) +b_\omega \Tr(\gamma)\end{equation}
whenever $\gamma\in X^+$. We define
\begin{equation}
    \cR := \left\{  \omega \in X^* \;\middle| \;\text{$\omega$ has $T$-bound $<1$}   \right\} \subset X^*. \label{eq:cR definition}
\end{equation}
Note that if $\omega\in\cR$, then we have
\begin{equation}\label{Equation: If relatively bounded then lower semibounded}
    \Tr(T\gamma)+\omega(\gamma) \geq (1-a_\omega)\Tr(T\gamma)-b_\omega\Tr(\gamma)\geq -b_\omega\Tr(\gamma) \text{ for all }\gamma\in X^+,
\end{equation}
and we say that $\Tr(T\cdot)+\omega $ is lower semibounded with lower bound $-b_\omega$. 
\begin{lemma}\label{Lemma: R is open}
    $\cR$ is a convex open neighborhood of $0\in X^*$, and as a consequence, for every $\omega\in X^*$ we have $\lambda \omega\in \cR$ for small enough $\lambda\in \R^+$.
\end{lemma}
\begin{proof}
    We first show that $\cR$ is convex. Let $\omega,\omega'\in \cR$, let $0\leq \alpha\leq 1$ and let $\tilde\omega:=\alpha \omega +(1-\alpha)\omega'$. Then, 
    \begin{equation*}
    \begin{split}
    \tilde\omega(\gamma)&=\alpha \omega(\gamma) +(1-\alpha)\omega'(\gamma) \\&\leq (\alpha a_{\omega} + (1-\alpha)a_{\omega'})\Tr(T\gamma)+ (\alpha b_{\omega} + (1-\alpha)b_{\omega'})\Tr(T\gamma) \\&\leq \max(a_{\omega},a_{\omega'})\Tr(T\gamma)+\max(b_{\omega},b_{\omega'})\Tr(\gamma),
    \end{split}
    \end{equation*} which satisfies $\tilde\omega\in \cR$ since $\max(a_{\omega},a_{\omega'}) < 1$. Hence, $\cR$ is convex. Next, we show that $\cR$ is open. Suppose that $\omega\in \cR$, let $0<\epsilon <1-a_\omega$ and let $\omega'\in B_{\|\cdot\|_{X^*}}(0,\epsilon)$. Then, if $\gamma\in X^+$, we have
    \begin{equation*}
        \begin{split}
    |(\omega + \omega')(\gamma)| &\leq a_\omega\Tr(T\gamma)+b_\omega\Tr(\gamma) + \|\omega'\|_{X^*}\Tr((1+T)\gamma) \\&\leq (a_\omega+\|\omega'\|_{X^*})\Tr(T\gamma) + (b_\omega+\|\omega'\|_{X^*})\Tr(\gamma) 
    \\&<(a_\omega+\epsilon)\Tr(T\gamma) + (b_\omega+\epsilon)\Tr(\gamma).
        \end{split}
    \end{equation*}
    Thus, since $a_\omega+\epsilon <1$, it follows that $\omega+\omega'\in\cR$ for all $\omega'\in B_{\|\cdot\|_{X^*}}(\epsilon)$. Hence, $\cR$ is open. 
\end{proof}
We now show that the KLMN Theorem (Theorem \ref{Theorem: KLMN}) allows us to represent the elements of $X^*$ in terms of Hermitian forms. Recall that $\cF(T)$ denotes the set of relatively $T$-bounded sesquilinear forms, a normed space per \eqref{Equation: norm on cF(T)}.
\begin{proposition} \label{Proposition: one-to-one correspondence}
    There is a one-to-one correspondence between $\cF(T)$ and $X^*$. This one-to-one correspondence is an isometry with respect to the norms $\|\cdot\|_{\cF(T)}$ and $\|\cdot\|_{X^*}$.  
\end{proposition}
\begin{proof}
    We first show how elements of $X^*$ are constructed from $\cF(T)$. Let $v\in \cF(T)$. For small enough $\lambda$, $\lambda v$ has $T$-bound less than $1$, and hence, by the KLMN-theorem the form sum $T\dot + \lambda v$ is self-adjoint and lower semibounded.     Next, for $\gamma\in X$, define 
    \begin{equation}
    \omega_v(\gamma) := \lambda^{-1}[\Tr((T\dot + \lambda v -\mu )\gamma) - \Tr(T\gamma)+\mu\Tr(\gamma)], \label{Equation: representation omega_v}
    \end{equation} where $\mu$ is the lower bound of $T\dot+\lambda v$. Notice that as consequence of \eqref{Equation: KLMN}, we have that 
    \begin{equation} \label{Equation: form is equal to dual element on the diagonal}
    \omega_v(|\psi\rangle \langle \psi|)=v(\psi,\psi) \text{ for all }\psi\in Q(T). 
    \end{equation} 
If $\gamma\in X^+$, we then have by Lemma \ref{Lemma: the diagonalization is in D(T^(1/2))} that
    \[|\omega_v(\gamma)|\leq\sum_n \lambda_n|v(\psi_n,\psi_n)|\leq C_v \sum_n \lambda_n \|\psi_n\|^2_T = C_v(\Tr((1+T)\gamma) = C_v\|\gamma\|_X.\]
    It follows from Lemma \ref{Lemma: Unique linear extension} that $\omega_v$ is bounded, i.e., $\omega_v\in X^*$. It  now follows that
    \begin{equation}
        \omega_v(\gamma) = \sum_n \lambda_nv(\psi_n,\psi_n) \text{ for all }\gamma\in X^+.\label{Equation: sum of forms}
    \end{equation}  
Next, we show how elements of $\cF(T)$ are constructed from $X^*$. Suppose that $\omega\in X^*$, and define
\begin{equation}\label{Hermitian form}
    v_\omega(\psi,\varphi) := \frac{1}{4}\sum_{k=0}^3 i^k\omega(|\psi+(-i)^k\varphi\rangle\langle \psi+(-i)^k\varphi|).
\end{equation} One can check that 
   \begin{equation}\label{Equation: form equals phi}
       v_\omega(\psi, \psi)=\omega(|\psi\rangle \langle \psi|)\text{ for all }\psi\in Q(T). 
   \end{equation}
    It is easy to check that $v$ is Hermitian sesquilinear and that $|v_\omega(\psi,\psi)|\leq C_\omega \|\psi\|_T^2$.  
    Hence $v_\omega\in \cF(T)$. 
    
    We next show that $v_{\omega_v}=v$ and $\omega_{v_\omega}=\omega$. For the first equality, notice that for all $\psi,\varphi\in Q(T)$, we have by \eqref{Equation: form is equal to dual element on the diagonal} that
    \[\omega_v(|\psi+(-i)^k\varphi\rangle \langle \psi + (-i)^k\varphi|) = v(\psi + (-i)^k\varphi,\psi + (-i)^k\varphi).\]
    We thus get
    \[v_{\omega_v}(\psi,\varphi)=\frac{1}{4}\sum_{k=0}^3 i^kv(\psi + (-i)^k\varphi,\psi + (-i)^k\varphi),\] which equals $v(\psi,\varphi)$ by the polarization identity. For the second equality, let $\omega\in X^*$, and let $\gamma\in X^+$. By \eqref{Equation: sum of forms}, we have 
    \[\omega_{v_\omega}(\gamma) = \sum_n \lambda_nv_\omega(\psi_n,\psi_n).\] By \eqref{Equation: form equals phi},  \[\sum_n \lambda_nv_\omega(\psi_n,\psi_n) = \omega(\gamma).\] We thus have $\omega_{v_\omega}(\gamma)=\omega(\gamma)$ for all $\gamma\in X^+$. By Lemma \ref{Lemma: Unique linear extension}, it follows that $\omega_{v_\omega}(\gamma)=\omega(\gamma)$ for all $\gamma\in X$. 

    Finally, we show that norms are preserved under this one-to-one correspondence. Let $v\in \cF(T)$. Then, 
    \[\|\omega_v\|_{X^*} = \sup_{\substack{\gamma\in X^+ \\\|\gamma\|_X=1}} |\omega_v(\gamma)| \geq  \sup_{\substack{\psi\in Q(T)\\ \|\psi\|_T=1}}|\omega_v(|\psi\rangle\langle \psi|)|=\sup_{\substack{\psi\in Q(T)\\ \|\psi\|_T=1}}|v(\psi,\psi)| =\|v\|_{\cF(T)}.\]
    Conversely, if $\omega\in X^*$ and $\gamma\in X^+$, then we have \[|\omega(\gamma)| = \left| \sum_n \lambda_n \omega(|\psi_n\rangle\langle \psi_n|)\right|=  \left| \sum_n \lambda_n v_\omega(\psi_n, \psi_n)\right|\leq \sum_n\lambda_n|v(\psi_n,\psi_n)|\]\[\leq \|v_\omega\|_{\cF(T)}\sum_n\lambda_n\leq \|v_\omega\|_{\cF(T)}\|\gamma\|_X.\]
   Dividing by $\|\gamma\|_X$ and taking the supremum over all $\gamma\in X^+$, it follows by Lemma \ref{Lemma: Unique linear extension} that $\|v_\omega\|_{\cF(T)}\geq \|\omega\|_{X^*}$. 
\end{proof}

\begin{remark}
    Because of the one-to-one correspondence in Proposition \ref{Proposition: one-to-one correspondence}, we shall identify elements of $\cF(T)$ and $X^*$. We thus regard any element of $X^*$ as representing a sesquilinear form through the relation $v(\psi,\varphi)\equiv v(|\varphi\rangle \langle \psi|)$. 
\end{remark}

\begin{remark}\label{Remark: Positive operator representation}
Elements of $X^*$ are \emph{a priori} only defined as linear functionals on $X$, which by Proposition \ref{Proposition: one-to-one correspondence} correspond to sesquilinear forms with domain $Q(T)$. There is in general no canonical way of representing elements of $X^*$ as operators on $\cH$. However, if $v\in \cR$, we have that $\Tr(T\,\cdot\,) + v$ is lower semibounded, and so the KLMN Theorem~\ref{Theorem: KLMN} tells us that there exists a unique lower semibounded self-adjoint operator $H_0(v)$ satisfying $\Tr(H_0(v) \gamma)=\Tr(T\gamma)+v(\gamma)$ for all $\gamma\in X$. The operator $H_0(v)$ is then given by the form sum $T\dot + v$. We thus regard $\cR$ as a \emph{generalized set of potentials} for our theory consisting of elements of $X^*$ giving rise to lower-semibounded Hamiltonians. 
\end{remark}

Next, we introduce the notion of infinitesimal bounds. We say that $w\in X^*$ has \emph{infinitesimal $T$-bound} if it has $T$-bound $\epsilon$ for all $\epsilon >0$\footnote{This property is also called ``$w$ is relatively form $T$-bounded with relative $T$-bound zero'' \cite[Section 10.7]{Schmudgen2012}}. Note that if $w$ has infinitesimal $T$-bound, then $w\in \cR$. We denote the set of infinitesimally $T$-bounded elements of $X^*$ by $\cR_\epsilon\subset \cR$. It is not hard to show that $\cR_\epsilon$ is a linear subspace of $X^*$. Interestingly, $\cR_\epsilon$ is closed in $X^*$ and is hence itself a Banach space with respect to the $X^*$-norm.

\begin{lemma}\label{Lemma: R inf is closed}
    $\cR_\epsilon$ is closed in $X^*$. 
\end{lemma}
\begin{proof}
    Let $\{w_n\}_{n\in \N}$ be a sequence in $\cR_\epsilon$, that converges in norm to $w\in X^*$. We will show that $w\in \cR_\epsilon$. Let $\epsilon >0$. 
    Since $w_n\in \cR_\epsilon$, there exists a $b_n>0$ such that 
    \[|w_n(\gamma)|\leq \frac{\epsilon}{2}\Tr(T\gamma) + b_n\Tr(\gamma) \text{ for all }\gamma\in X^+. \] Furthermore, we have that
    \[|w(\gamma)|\leq |w_n(\gamma)| + |w_n(\gamma)-w(\gamma)|\leq \frac{\epsilon}{2}\Tr(T\gamma) + b_n\Tr(\gamma) + \|w_n-w\|_{X^*}\|\gamma\|_X\] for all $\gamma\in X^+$. Since $w_n\to w$, we have for large enough $n$ that $\|w_n-w\|_{X^*}<\epsilon/2$. Hence, using the fact that $\|\gamma\|_X = \Tr((1+T)\gamma)$ for positive $\gamma$, we get
\[|w(\gamma)|\leq \epsilon\Tr(T\gamma) + \left(b_n+\frac{\epsilon}{2}\right)\Tr(\gamma)\] for all $\gamma\in X^+$, which is what we wanted to show. 
\end{proof}

\begin{lemma}\label{Lemma: R and R inf}
Let $w \in \cR_\epsilon$ and $v\in\cR$. Then, the following hold: 
    \begin{enumerate}
        \item $v+w$ has $T$-bound $<1$.
        \item $v$ has $T\dot +w$-bound $<1$.
        \item $w$ has infinitesimal $T\dot +v$-bound.
    \end{enumerate}
    Furthermore, we have
    \[T\dot + (w + v) = (T\dot +v)\dot +w = (T\dot +w)\dot+v,\]
    and this operator is lower semibounded. 
\end{lemma}
Notice that point 1.~means precisely that $\cR+\cR_\epsilon=\cR$. Furthermore, one can show that $\cR_\epsilon$ consists precisely of those elements $w$ of $\cR$ for which $\lambda w\in \cR$ for all $\lambda\in\R$. 
\begin{proof}
\noindent{\emph{Proof of 1.:}} Let $a_v<1$ denote the $T$-bound of $v$, and pick $\epsilon = (1-a_v)/2$. Then, since $w$ has $T$-bound $\epsilon$, we have
\[|(v+w)(\gamma)|\leq (a_v + (1-a_v)/2)\Tr(T\gamma) + (b_v+b_w)\Tr(\gamma)\]
for all $\gamma\in X^+$. Hence, $v+w$ has $T$-bound $(1+a_v)/2 < 1$. 

\noindent{\emph{Proof of 2.:}} Pick $0<\epsilon<1-a_v$, such that $a_v/(1-\epsilon)<1$. Since $w$ has $T$-bound $\epsilon$, we have \[-b_w\Tr(\gamma) -\epsilon\Tr(T\gamma)\leq w(\gamma) \text{ for all } \gamma\in X^+ \]
\[\Longleftrightarrow\]
\[ (1-\epsilon)\Tr(T\gamma)\leq \Tr((T\dot+w)(\gamma) )+b_w\Tr(\gamma)\text{ for all } \gamma\in X^+. \]
It follows that 
\begin{align*}
    |v(\gamma)| &\leq a_v\Tr(T\gamma) + b_v\Tr(\gamma) \\
    &=\frac{a_v}{1-\epsilon}(1-\epsilon)\Tr(T\gamma) + b_v\Tr(\gamma)\\
    &\leq \frac{a_v}{1-\epsilon}\Tr((T\dot +w)\gamma) + (b_v+b_w)\Tr(\gamma)
\end{align*}
for all $\gamma\in X^+$. Hence, $v$ has $T\dot +w$-bound $a_v/(1-\epsilon)<1$.

\noindent{\emph{Proof of 3.:}} Let $M>0$ and pick $0<\epsilon\leq(1-a_v)/M$. Since $w$ has $T$-bound $\epsilon$, we have 
\[|w(\gamma)|\leq \epsilon \Tr(T\gamma) + b_w\Tr(\gamma)\text{ for all }\gamma\in X^+.\]
Since $v$ has $T$-bound $a_v <1$, we have
\[-b_v\Tr(\gamma) -a_v\Tr(T\gamma)\leq v(\gamma) \text{ for all }\gamma\in X^+\]
\[\Longleftrightarrow\]
\[(1-a_v)\Tr(T\gamma)\leq \Tr((T\dot +v)\gamma) + b_v\Tr(\gamma)\text{ for all }\gamma\in X^+.\]
It follows that
\begin{align*}
    |w(\gamma)|&\leq \epsilon \Tr(T\gamma) + b_w\Tr(\gamma) \\ 
    &\leq\frac{1-a_v}{M}\Tr(T\gamma)+b_w\Tr(\gamma)\\
    &\leq \frac{1}{M}\Tr((T\dot +v)\gamma) + (b_w+b_v)\Tr(\gamma)
\end{align*}
for all $\gamma\in X^+$. Since $M>0$ was arbitrary, it follows that $w$ has infinitesimal $T\dot + v$-bound. 

The KLMN theorem now gives three self-adjoint operators $H_1 = T \dot{+} (w + v)$, $H_2 = (T \dot{+} w) \dot{+}v$, and $H_3 = (T \dot{+} v)\dot{+}w$, which all three coincide, since they coincide with the unique self-adjoint operator associated with the closed form $h = t +v + w$ with domain $Q(T)$. Furthermore, since $v+w\in \cR$, it follows from equation \ref{Equation: If relatively bounded then lower semibounded}, that $h$ is lower semibounded. 
\end{proof}

\begin{example} \label{Example: The coulomb example !}
    Suppose $\cH=L^2(\R^3)$, $T = -\nabla^2$ and let $v$ denote the sesquilinear form \[
    v(\psi,\varphi)=-Z\int dx \frac{\overline{\psi(x)}\varphi(x)}{|x|}. \]
    The operator $T$ has form domain $H^{1}(\R^3)$, and it is well-known that $v$ has infinitesimal $T$-bound, meaning that $v\in \cR_\epsilon$ \cite{katoFundamentalPropertiesHamiltonian1951}. Thus, the Hamiltonian $H_0(v) = T \dot{+} v$ is self-adjoint with form domain $H^{1}(\R^3\times\Z_2)$. It represents the element $\Tr(H_0(v)\, \cdot) = \Tr(T\,\cdot)+ v\in X^*$
\end{example}

We end this section with a technical lemma which will be useful for later proofs. 

\begin{lemma}\label{Lemma: v is lower semicontinuous}
    Let $v\in (X^*)^+$, and let $\{\gamma_n\}_{n\in \N}$ be a sequence in $X^+$ which is bounded in the $X$-norm and satisfies ${\gamma_n\xrightharpoonup{*}\gamma\in X^+}$. Then, 
    \[v(\gamma)\leq \liminf_{n\to\infty}v(\gamma_n)\]
\end{lemma}

\begin{proof}
    Define $U:Q(T)\to \cH$ by $U\psi:=(1+T)^{1/2}\psi$. This is well-defined, since $\|U\psi\|^2=\|\psi\|_T^2<\infty$ for all $\psi\in Q(T)$. Consider the sequence $\{\tilde\gamma_n\}_{n\in\N}$, defined by
    \[\tilde\gamma_n := U\gamma_n U.\]
    $\tilde\gamma_n$ is trace-class for every $n$ since $\Tr(\tilde\gamma_n)=\Tr((1+T)\gamma_n)<\infty$. Since $\{\gamma_n\}_{n\in\N}$ is bounded in the $X$-norm, $\{\tilde\gamma_n\}_{n\in\N}$ is bounded in trace-norm and hence has a weak-$*$ convergent subsequence by the Banach-Alaoglu theorem. Restrict  $\{\tilde\gamma_n\}_{n\in\N}$ to this subsequence and denote the weak-$*$ limit by $\tilde\gamma$. For any $\varphi,\psi\in Q(T)$, we have \[\langle \varphi, U\gamma U \psi \rangle = \lim_{n\to\infty}\langle \varphi, U\gamma_n U \psi \rangle = \lim_{n\to\infty}\langle \varphi, \tilde\gamma_n\psi\rangle = \langle \varphi,\tilde\gamma\psi\rangle,  \] implying that $\tilde\gamma = U\gamma U$. Hence, we have $\tilde\gamma_n\xrightharpoonup{*}\tilde\gamma = U\gamma U$.

    Next, define a sesquilinear form $\tilde v$ on $\cH$ by \[\tilde v(\varphi,\psi) := v(U^{-1}\varphi,U^{-1}\psi). \]
    Then, $\tilde v$ is bounded since \[|\tilde v(\varphi,\psi)|\leq C_v\|U^{-1}\varphi\|_T\|U^{-1}\psi\|_T = C_v\|\varphi\|\|\psi\|. \] Hence, there exists a bounded operator $B\in \cB^+(\cH)$ satisfying 
    \[\langle \varphi,B\psi\rangle = \tilde v(\varphi,\psi)\] for all $\varphi,\psi\in \cH$. We will show that  \begin{equation}\label{Equation: v = tr B}
v(\gamma)=\Tr(B\tilde\gamma) \text{ and }v(\gamma_n)=\Tr(B\tilde \gamma_n) \text{ for all } n\in\N.
   \end{equation} Assuming that \eqref{Equation: v = tr B} holds, we can conclude by Lemma \ref{Lemma: Trace of unbounded operator is lower semicontinuous} that
    \[v(\gamma)=\Tr(B\tilde\gamma )\leq \liminf_{n\to\infty}\Tr(B\tilde\gamma_n) = \liminf_{n\to\infty}v(\gamma_n).\]
It remains to prove \eqref{Equation: v = tr B}. Let $\gamma\in X^+$ be arbitrary, and let $\gamma = \sum_k\lambda_k|\psi_k\rangle\langle \psi_k|$ be its spectral decomposition. By Lemma \ref{Lemma: the diagonalization is in D(T^(1/2))}, we know that $\psi_k\in Q(T)$ for all $k$. As a consequence, the sum $\sum_k\lambda_k|\psi_k\rangle\langle \psi_k|$ converges to $\gamma$ in the $X$-norm: For any partial sum $\gamma_m:= \sum_{k=1}^m\lambda_k|\psi_k\rangle \langle \psi_k|$, we have
\[\|\gamma-\gamma_m\|_X = \left \|\sum_{k=m+1}^\infty \lambda_k|\psi_k\rangle \langle \psi_k| \right\|_X = \sum_{k=m+1}^\infty \lambda_k\|\psi_k\|_T^2\xrightarrow{m\to \infty} 0.\] Since $\gamma\mapsto v(\gamma)$ is continuous in the $X$-norm, we have
\[v(\gamma)=\sum_k\lambda_kv(\psi_k,\psi_k)=\sum_k\lambda_k\langle U\psi_k,BU\psi_k\rangle = \sum_k \lambda_k\Tr(B U|\psi_k\rangle\langle \psi_k|U).\] 
For any decomposition $\gamma = \gamma^+-\gamma^-$, we have
\[\|U\gamma U\|_1 \leq \|U\gamma^+ U\|_1 + \|U\gamma^- U\|_1 = \Tr((1+T)\gamma^+) + \Tr((1+T)\gamma^-),\]
implying $\|U\gamma U\|_1 \leq \|\gamma\|_X$, which means that $(\gamma\mapsto U\gamma U):  X \to \cS_1(\cH)$ is bounded. It follows that $\Tr(B \; \cdot\;):X \to \R$ is continuous, so \[\sum_k \lambda_k\Tr(B U|\psi_k\rangle\langle \psi_k|U) = \Tr(B (U\gamma U)).\]
We hence have $v(\gamma)=\Tr(B (U\gamma U))$ for all $\gamma\in X^+$, and so equation \eqref{Equation: v = tr B} holds in particular. 
\end{proof}

\section{Surjectivity of the partial trace map}\label{Section: Surjectivity of the partial trace map}

In this section, we introduce the $N$-particle equivalent of $X$, denoted by $X^N$, which is a Banach space containing the \emph{($N$-particle) density matrices}. Furthermore, we prove the surjectivity and boundedness of the partial trace map $\Pi$ introduced in Section \ref{Section: Preliminaries} as a map from $X^N$ to $X$. 

Let $\hat{T}$ denote the second-quantization of $T$ introduced in Section \ref{Section: Preliminaries} and define
\[X^N := \left\{\Gamma^+-\Gamma^- \;\middle|\; \Gamma^\pm\in \cS_1^+(\cH^N), \Tr(\hat T\Gamma^\pm)<+\infty\right\},\]
equipped with the norm 
\[\|\Gamma\|_{X^N} :=   \inf \left\{ \Tr((1+\hat T)\Gamma^+)+\Tr((1+\hat T)\Gamma^-) \;\middle|\; \Gamma = \Gamma^+ - \Gamma^-, \Gamma^\pm \geq 0\right\}.\]

In an analogous manner as for $X$, the space $X^N$ is a real Banach space, and its dual space is given by  
\[(X^N)^* := \left\{\hat \omega: X^N\to \R\;\middle|\; \hat \omega \text{ is linear and } \exists C>0 : |\hat \omega(\Gamma)|\leq C \|\Gamma\|_{X^N}  \text{ for all } \Gamma\in X^N\right\}.\] Analogously as in Proposition \ref{Proposition: one-to-one correspondence}, we have the isometric equivalence \[(X^N)^*\cong \left\{\hat v : Q(\hat T)\times Q(\hat T)\to \C \text{ Hermitian } \;\middle|\;\exists C_{\hat v}: \forall \Psi\in Q(\hat T): |\hat v(\Psi,\Psi)|\leq C_{\hat v}\|\Psi\|^2_{Q(\hat T)}\right\},\] the right hand side of which we also denote by $\cF(\hat T)$. The $N$-particle generalization of $\cR\subset X^*$ is defined as
\begin{equation}
    \cR^N := \{ \hat v \in (X^N)^* \mid \text{$\hat v$ has $\hat{T}$-bound $<1$} \} \label{eq: definition of cRN},
\end{equation}
while the $N$-particle generalization of $\cR_\epsilon\subset \cR$ is defined as 
\[\cR_\epsilon^N := \{\hat v \in (X^N)^* \mid \text{$\hat v$ has infinitesimal $\hat{T}$-bound} \}. \]
We introduce the convex sets
\begin{align}\begin{split}
    \cC &:= \{\gamma\in \cS_1(\cH)| \Tr(\gamma)=N, 0\leq \gamma\leq 1\},\\
    \cD &:= \cC\cap X,\\
    \cC^N &:= \{\gamma\in \cS_1(\cH^N)| \Tr(\gamma)=1, 0\leq \gamma\leq 1\},\\
    \cD^N &:= \cC^N\cap X^N.\label{eq: cD definition}
\end{split}\end{align}
The sets $\cC^N$ and $\cC$ contain, respectively, all mixed $N$-particle states and all corresponding reduced density matrices. By Lemma~\ref{Lemma: the diagonalization is in D(T^(1/2))}, the set $\cD^N$ consists of precisely finite-kinetic energy mixed $N$-particle states. The extremal points of this set are pure-state finite-kinetic energy density matrices, i.e., if $\Psi\in Q(\hat T)$, then $|\Psi\rangle\langle \Psi|\in \cD^N$. Finally, we regard $\cD$ as the set of reduced density matrices, which consists precisely of finite-kinetic energy mixed one-particle states.
\begin{theorem}\label{Theorem: Tr_{-1} is surjective}
The partial trace map $\Pi:\cS_1^\mathrm{sa}(\cH^N)\to \cS_1^\mathrm{sa}(\cH)$ is surjective, positive, satisfies $\|\Pi(\Gamma)\|_1 \leq N\|\Gamma\|_1$, and the following: \begin{enumerate}
    \item $\Pi$ maps $\cC^N$ surjectively onto $\cC$,
    \item $\Pi$ maps $\cD^N$ surjectively onto $\cD$,
    \item $\|\Pi(\Gamma)\|_{X}\leq N\|\Gamma\|_{X^N}$ for all $\Gamma\in X^N$. 
    \item $\Pi(X^N)$ is dense in $X$. 
\end{enumerate}
\end{theorem}

\begin{proof}
Let $\cK:=\mathcal H^{N-1}$. We can identify $\mathcal H^N$ with a closed subspace of $\mathcal H\otimes \cK$. The partial trace over $\cK$ is the unique linear map
\[
\Tr_\cK:\; \mathcal S_1(\mathcal H\otimes \cK)\to\mathcal S_1(\mathcal H)
\]
characterized by
\[
\Tr\big(V\,\Tr_\cK(\Gamma)\big) = \Tr\big((V\otimes \mathbbm{1}_\cK)\,\Gamma\big),
\qquad
\forall\, V\in\mathcal B(\mathcal H),\; \Gamma\in\mathcal S_1(\mathcal H\otimes \cK).
\]
Restricting this map to $\mathcal S_1^\mathrm{sa}(\mathcal H^N)\subset \mathcal S_1(\mathcal H\otimes \cK)$ gives the operator $\Pi := N\Tr_{\cK}\restriction_{\cS^{\mathrm{sa}}_1(\cH^N)}$.

We first prove that $\Pi$ is positive. For this, it is enough to show that $\Tr_{\cK} : \cS_1(\cH \otimes \cK) \to \cS_1(\cH)$ is positive. Suppose $\Gamma\in \cS_1^+(\cH \otimes \cK)$. Then for every $\varphi\in\mathcal H$,
\[
\langle \varphi,\,\Tr_\cK(\Gamma)\,\varphi\rangle
= \Tr\!\big((|\varphi\rangle\langle\varphi|\otimes \mathbbm{1}_\cK)\,\Gamma\big)\ge0,
\]
because $|\phi\rangle\langle\phi|\otimes \mathbbm{1}_\cK\ge0$. Hence $\Tr_\cK(\Gamma)\ge0$. Thus $\Pi$ is positive.

We next show $\|\Pi(\Gamma)\|_1 \leq N\|\Gamma\|_1$.
By the duality between $\mathcal S_1$ and $\cB(\cH)$,
\begin{align*}
\|\Tr_\cK(\Gamma)\|_1
&= \sup_{\|B\|\le1}
\big| \Tr(B\,\Tr_\cK(\Gamma))\big| \\
&= \sup_{\|B\|\le1} \big|\Tr\big((B\otimes \mathbbm{1}_\cK)\Gamma\big)\big|,
\end{align*}
where the supremum is taken over $B \in \cB(\cH)$ in both cases.
Since $\|B\otimes \mathbbm{1}_\cK\|=\|B\|\le1$, each term is bounded by $\|\Gamma\|_1$. Taking the supremum yields
\[
\|\Tr_\cK(\Gamma)\|_1 \le \|\Gamma\|_1.
\]
Hence, $\|\Pi(\Gamma)\|_1\leq N\|\Gamma\|_1$. 

In order to prove surjectivity, we need the two following technical lemmas: 
\begin{lemma}\label{lemma: affine combinations in D}
For every $\psi \in \cH$ with $\|\psi\|=1$, there exist rank-$N$ projections $q_0,q_1,\cdots q_{N-1},r_1,\cdots, r_{N-1}$ such that
\[
    N |\psi\rangle\langle \psi|  = q_0 + \sum_{k=1}^{N-1} q_k-r_k.
\]
\end{lemma}
\begin{proof}
Let $\psi_1 := \psi$, and let $\{\psi_n\}_{n\in\N}$ be a completion of $\psi$ to an orthonormal basis of $\cH$. Denote $P[n]:= |\psi_n\rangle \langle \psi_n|$, and $P[N,M] = P[N]+P[N+1]+\cdots + P[M-1]  +P[M]$ whenever $M>N$. Notice that $\Tr(P[N,M])= M-N +1. $
    Next, pick
    \begin{align*}
        q_0 &:= P[1,N], \\
        q_k & := P[1] + P[k(N-1)+2,(k+1)N -k] \text{, and } \\
        r_k &:= P[(k-1)N+2,kN+1]
    \end{align*}
for all $k=1,2,3,\cdots N-1$. We see that $q_0, q_k$ and $r_k$ are rank-$N$ projections for every $k$. Notice that 
\begin{equation}\label{Equation: q equals sum of Ps}
    q_0 + \sum_{k=1}^{N-1} q_k = N P[1] + P[2, N^2-N+1],
\end{equation}
and 
\begin{equation}\label{Equation: r equals sum of Ps}
    \sum_{k=1}^{N-1}r_k = P[2,N^2-N+1].
\end{equation}
Hence, \begin{equation}\label{Equation: Projection telescope sum}
    N |\psi\rangle\langle\psi| = NP[1]= q_0 + \sum_{k=1}^{N-1} q_k-r_k.
\end{equation}

\begin{figure}[h]
    \centering
    \includegraphics[width=0.4\textwidth]{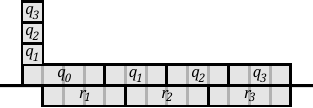}
    \caption{This figure illustrates the proof of Lemma \ref{lemma: affine combinations in D} for $N=4$. Each square represents a projection, starting from $P[1]$ at the leftmost side to $P[4^2 - 4 + 1]=P[13]$ on the rightmost side. The squares above and below the middle line represent positive and negative summands in \eqref{Equation: Projection telescope sum} respectively. We see that the sum mostly cancels and that we are left with $4P[1]$.}
    \label{Figure: Proof of sum of projections}
\end{figure}

\end{proof}
\begin{lemma}\label{Lemma: Unique slater determinant preimage}
    Let $\psi_1,\psi_2,...,\psi_N\in \cH$ be orthonormal. Then the rank-$N$ projection
    \[\gamma := |\psi_1\rangle\langle \psi_1|+\cdots +|\psi_N\rangle\langle \psi_N|\in \cC\] has a unique preimage in $\cC^N$ under $\Pi$, which is given by $$\Gamma = |\psi_1\wedge\cdots\wedge\psi_N\rangle\langle\psi_1\wedge\cdots\wedge\psi_N|.$$
\end{lemma}
\begin{proof} This was shown by Coleman \cite[Section 9]{Coleman1963}.
\end{proof}

We are ready to prove that $\Pi:\cS_1^\mathrm{sa}(\cH^N)\to \cS_1^\mathrm{sa}(\cH)$ is surjective. Let $\gamma\in \cS_1^+(\cH)$, and let 
\[\gamma = \sum_{n=1}^\infty \lambda_n|\psi_n\rangle \langle \psi_n|\] be its spectral decomposition. Let 
\[N|\psi_n\rangle\langle \psi_n| = q_0^n + \sum_{k=1}^{N-1}q_k^n -r_k^n\]
denote the rank-$N$ projections given in Lemma \ref{lemma: affine combinations in D}, and let $Q_0^n,Q_1^n,\cdots Q_{N-1}^n, R_1^n,\cdots R_{N-1}^n$ denote the corresponding Slater determinants given in Lemma \ref{Lemma: Unique slater determinant preimage}. Then, define \begin{equation}\label{Equation: Gamma preimage in terms of slater determinants}
    \Gamma := \frac{1}{N} \sum_{n=1}^\infty \lambda_n\left(Q_0^n + \sum_{k=1}^{N-1}Q_k^n -R_k^n\right).
\end{equation}
This indeed defines an element of $\cS_1(\cH^N)$ since the sum is absolutely convergent in the trace-norm: 
\[\|\Gamma\|_1\leq \frac{1}{N}\sum_{n=1}^\infty \lambda_n\left(\|Q_0^n\|_1 + \sum_{k=1}^{N-1}\|Q_k^n\|_1 + \|R_k^n\|_1\right) = \frac{2N-1}{N}\sum_{n=1}^\infty\lambda_n =\frac{2N-1}{N}\|\gamma\|_1. \] 
Furthermore, we have
\[\Pi(\Gamma) = \frac{1}{N}\sum_{n=1}^\infty \lambda_n\left(q_0^n + \sum_{k=1}^{N-1}q_k^n -r_k^n\right) = \sum_{n=1}^\infty \lambda_n|\psi_n\rangle\langle \psi_n| = \gamma. \]
Next, suppose that $\gamma\in \cS_1^\mathrm{sa}(\cH)$ is not necessarily positive. Then, for any decomposition $\gamma = \gamma^+ -\gamma^- $ with $\gamma^\pm\in \cS_1^+(\cH)$, we can find preimages $\Gamma^\pm\in \cS_1^\mathrm{sa}(\cH^N)$ such that $\Pi(\Gamma^+-\Gamma^-)=\gamma^+-\gamma^- = \gamma$. Hence, $\Pi$ is surjective. We now proceed with proving Points 1--4 of the theorem:

\noindent{\emph{Proof of 1.}} was proven by Coleman in 1963 \cite{Coleman1963}. 

\noindent{\emph{Proof of 2.}} Let $\gamma\in \cD$.  Since $\cD\subset \cC$, we can find some $\Gamma\in \cC^N$ by point (1) such that $\Pi(\Gamma)=\gamma$. It follows from Lemmas \ref{Lemma: second quantization expectation value} and \ref{Lemma: Trace inequality} and by positivity of $\gamma$ and $\Gamma$ that
\begin{equation}\label{Equation: Positive norm equality}
    \|\Gamma\|_{X^N} = \Tr(\Gamma) + \Tr\big(\hat T\Gamma\big)=\frac{1}{N}\Tr(\gamma) + \Tr(T\gamma)\leq \|\gamma\|_X,
\end{equation} for any $\Gamma\in\Pi^{-1}(\gamma)\cap \cC^N$. This proves that $\Gamma\in \cD^N$, and hence that $\Pi$ maps $\cD^N$ surjectively onto $ \cD$. 

\noindent{\emph{Proof of 3.}} Let $\Gamma\in X^N$, and $\gamma:=\Pi(\Gamma)$, and let $\Gamma= \Gamma^+ -\Gamma^-$ be any decomposition with $\Gamma^\pm\in X^+$. Denote $\gamma^\pm := \Pi(\Gamma^\pm)$. Then, 
\[\|\gamma\|_X\leq \Tr\big((1+T)\gamma^+\big)+\Tr\big((1+T)\gamma^-\big) \leq N\Tr\big((1+\hat T)\Gamma^+\big)+N\Tr\big((1+\hat T)\Gamma^-\big). \]
Taking the infimum over decompositions, we get \begin{equation}\label{Equation: gamma norm is smaller than Gamma norm}
    \|\gamma\|_X\leq N\|\Gamma\|_{X^N},
\end{equation} 
implying that $\gamma\in X$, and also point 3. 

\noindent{\emph{Proof of 4.}} It is easy to show that any finite-rank operator in $X$ has a preimage in $X^N$ under $\Pi$. Indeed, if $\gamma\in X^+$ has finite rank, then using the construction \eqref{Equation: Gamma preimage in terms of slater determinants}, we get 
\[\|\Gamma\|_{X^N}\leq \frac{1}{N}\sum_{n=1}^r \lambda_n\left(\|Q_0^n\|_X + \sum_{k=1}^{N-1}\|Q_k^n\|_X + \|R_k^n\|_X\right).\] We can pick a basis for this construction such that each of the terms is finite, and hence $\|\Gamma\|_{X^N}<+\infty$ since the sum is finite. The fact that $\Pi(\Gamma)$ is dense then follows from the fact that finite-rank operators are dene in $X$.
\end{proof}

\begin{remark}\label{Remark: second quantization map for X*}
    The adjoint of $\Pi$ when regarded as a map $\Pi|_{X^N}:X^N\to X$ is a map $(\Pi|_{X^N})^* : X^* \to (X^N)^*$ which is injective. This follows from $\Pi|_{X^N}$ being bounded and $\Pi(X^N)$ being dense in $X$. When $\cB^\mathrm{sa}(\cH)$ is regarded as a subspace of $X^*$, $(\Pi|_{X^N})^*$ restricts to the second quantization map $\hat{(\cdot)}$ defined in \eqref{Equation: Second quantization}. We therefore keep the terminology \emph{second quantization map} for $(\Pi|_{X^N})^*$ and the notation $(\Pi|_{X^N})^*(v)=\hat v$ for $v\in X^*$. The second quantization map satisfies $\|\hat v\|_{(X^N)^*} \leq N\|v\|_{X^*}$ for all $v\in X^*$. 
\end{remark}

The significance of Theorem \ref{Theorem: Tr_{-1} is surjective} is that every one-particle reduced density matrix $\gamma$ represents a class of $N$-particle density matrices $\Gamma$ through the partial trace map. Furthermore, since $\Pi$ is bounded, the class $\Pi^{-1}(\gamma)$ is closed in $X^N$. 

\begin{remark}\label{Remark: relatively bounded sets}
As a consequence of Remark \ref{Remark: second quantization map for X*}, we have $v\in \cR \Leftrightarrow \hat v\in \cR^N$. Similarly, we have $v\in \cR_\epsilon\Leftrightarrow \hat v\in \cR_\epsilon^N$.
\end{remark}

\section{Reduced Density Matrix Functional Theory}\label{Section: Reduced Density Matrix Functional Theory}

In this section, we introduce Reduced Density Functional Theory (RDMFT) using the framework introduced in sections \ref{Section: Reduced Density Matrices}, \ref{Section: Single-particle Hamiltonians} and \ref{Section: Surjectivity of the partial trace map}. We prove that the universal density matrix functional is lower semicontinuous in this setting. 

For a system of $N$ interacting fermions in an external potential $V$, the total Hamiltonian is given by 
\begin{equation}
    \begin{split}
    \hat{H} &= \hat{H}_{\hat W} + \hat{V} \\ &=\hat{T} + \hat{W} + \hat{V},
    \end{split}
\end{equation}
where $\hat{T}$ is a fixed one-body operator, typically kinetic energy, and $\hat{w}$ is the inter-particle interaction potential, a two-or-higher body operator. For $\hat T$ and $\hat V$, the carets indicate that the operators are second quantizations of single-particle operators $T$ and $V$. Presently,  $\hat{H}$ is a \emph{formal} operator; the domain must be defined and self-adjointness must be proven. 
\begin{example}
    We exemplify this with the following Hamiltonian of a molecule in the clamped-nuclei approximation \cite{katoFundamentalPropertiesHamiltonian1951}: $\cH = L^2(\R^3 \times \Z_2)$, with
\begin{equation}\label{Equation: Hamiltonian example}
    \hat T = - \sum_i\nabla^2_i, \quad  \hat V =  -\sum_a \sum_i \frac{Z_a}{|R_a-x_i|}, \quad \text{and} \quad  \hat W = \sum_{i<j}\frac{1}{|x_i-x_j|},
\end{equation}
with the Hamiltonian $\hat{H}$ acting on sufficiently regular elements of the $N$-electron Hilbert space $\cH^N = \bigwedge^N L^2(\R^3\times \Z_2)$. In a famous paper~\cite{katoFundamentalPropertiesHamiltonian1951}, T.~Kato proved that $\hat{H}$ has a self-adjoint realization with domain $H^2(\R^{3N}\times \Z_2^N) \cap \bigwedge^N L^2(\R^3\times \Z_2)$ and form-domain
\[Q(\hat H) = H^1(\mathbb{R}^{3N}\times \Z_2^N) \cap \bigwedge^N L^2(\mathbb{R}^3\times \Z_2).\]
Key here was that $\hat{w}$ has infinitesimal $\hat{T}$-bound.

As for reduced density matrices, any $\gamma \in \cD \subset X \subset \cS_1^\text{sa}(L^2(\R^3\times\Z_2))$ is now a \emph{spin-density matrix} with finite kinetic energy, and elements $v \in \cR \subset X^*$ are intuitively spin dependent potentials that are not too singular, such that they have finite interaction with any $\gamma\in X$. The nuclear potential $V$ given above is an example; it has $T$-bound $<1$ which guarantees $\hat{H}$ to have finite ground-state energy by Lemma \ref{Lemma: R and R inf} and Remark \ref{Remark: relatively bounded sets}.
\end{example}

In RDMFT one studies the ground-state energy as function of $v$, which then is a parameter of the problem. We appeal to the KLMN theorem to define a class of parameter-dependent Hamiltonians $\hat{H}(v) \in (X^*)^N$.
For the remainder of this article, let $\hat w\in(X^*)^+$, and suppose that $\hat w$ is infinitesimally $T$-bounded. By Remark \ref{Remark: Positive operator representation}, there is a unique self-adjoint operator $\hat H_{\hat{w}}\geq 0$ defined by
$$
    \hat{H}_{\hat w} := \hat{T} \dot{+} \hat{w},
$$
the form sum of $\hat T$ and $\hat w$, which satisfies $Q(\hat H_{\hat{w}})=Q(\hat{T})$ and \[\Tr(\hat H_{\hat w}\Gamma)=\Tr(\hat T\Gamma)+\hat w(\Gamma)\] for all $\Gamma\in X^N$ (see \cite[Proposition 10.22]{Schmudgen2012}). By Lemma \ref{Lemma: R and R inf}, any $v\in \cR$, has $\hat T \dot + \hat w$-bound $<1$, which implies that the Hamiltonian
\[\hat H(v):= \hat H_{\hat w} \dot +\hat v\]
is well-defined as a lower-semibounded self-adjoint operator with form domain $Q(\hat T)$. 

We define the set of finite-kinetic energy wave functions as
\[\cW^N := \left\{\Psi \in \cH^N \middle| \|\Psi\|= 1, \|\hat T^{1/2}\Psi\|^2<+\infty\right\}. \]
The ground state energy map $E_\RDM: X^* \to \R \cup \{-\infty\}$ is defined by 
\[E_\RDM(v) :=   
    \inf_{\Psi \in \cW^N} \bigl(\hat h_{\hat w}( \Psi,\Psi) + \hat{v}(\Psi,\Psi) \bigr),
\]
where $\hat h_{\hat w}$ is the sesquilinear form associated to $\hat H_{\hat w}$. For $v \in \cR$, it follows by the variational characterization of the bottom of the spectrum of self-adjoint operators that $E_\RDM(v)=\inf\operatorname{spec}(\hat{H}(v))$. If $v \in X^*\setminus \cR$, the ground-state energy $E_\RDM(v)$ is still defined in the above variational sense, even if there is no self adjoint operator representing the form sum $\hat H_{\hat w} \dot{+} \hat v$. An alternative definition for $v\in \cR$ can be given as 
\[E_\RDM'(v) = 
    \inf_{\Gamma \in \cD^N}\Tr(\hat{H}(v)\Gamma) .
\]
$E_\RDM'$ and $E_\RDM$ coincide in $\cR$:
\begin{lemma}\label{Lemma: E=E'}
    $E_\RDM(v)=E_\RDM'(v)$ for all $v \in \cR$. 
\end{lemma}
\begin{proof}
     Indeed, let $v\in X^*$ and $\Psi\in Q(\hat T)$. Since $|\Psi\rangle \langle \Psi|\in \cD^N$, and $\Tr(\hat H(v) |\Psi\rangle\langle \Psi|)=\hat h_{\hat w}( \Psi,\Psi) + \hat{v}(\Psi,\Psi)$, we have $E_\RDM'(v)\leq E_\RDM(v)$ by definition. Since $\Gamma $ is a convex combination of pure states, we have by Lemma \ref{Lemma: the diagonalization is in D(T^(1/2))} and \eqref{Equation: Split the T} that
\[\Tr(\hat H(v)\Gamma) = \sum \Gamma_i \Tr(\hat H(v) |\Psi_i\rangle\langle \Psi_i|) \geq \inf_i \Tr(\hat H(v) |\Psi_i\rangle\langle \Psi_i|)\rangle\geq E_\RDM(v).\] Taking the infimum over all $\Gamma\in \cD^N$, we get $E_\RDM(v)\leq E_\RDM'(v)$. 
\end{proof}

For elements $v,v'\in X^*$ we write $v \geq v'$ if $v(\gamma) \geq v'(\gamma)$ for every $\gamma \in \cS^+_1(\cH)$.
The following lemma is analogous to Theorem 3.1 of \cite{Lieb1983}.
\begin{lemma}
We have the following facts about $E_\RDM$: 
    \begin{enumerate}
        \item $E_\RDM$ is concave, that is, $E_\RDM(v)\geq \alpha E_\RDM(v_1)+ (1-\alpha)E_\RDM(v_2)$ for \mbox{$v=\alpha v_1+(1-\alpha)v_2$,} $0\leq \alpha\leq 1$, and $v_1,v_2\in X^*$.
        \item $E_\RDM$ is monotone decreasing, that is, $E_\RDM(v_1)\geq E_\RDM(v_2)$ if $v_1\geq v_2$ 
        \item $E_\RDM(v)$ is finite for all $v\in \cR$
        \item $E_\RDM$ is upper-semicontinuous with respect to $\|\cdot\|_{X^*}$.
        \item $E_\RDM $ is locally Lipschitz in $\cR$. 
        \end{enumerate} 
\end{lemma}
\begin{proof}

\noindent{\emph{Proof of 1.}} Let $v_1,v_2\in X^*$ and $v=\alpha v_1+(1-\alpha)v_2$, $0\leq \alpha\leq 1$. We then have 
    \[\Tr(\hat H(V)\Gamma) = \alpha \Tr(\hat H(v_1)\Gamma) + (1-\alpha)\Tr(\hat H(v_2)\Gamma) \geq \alpha E_\RDM(v_1) + (1-\alpha)E_\RDM(v_2)\] for all $\Gamma \in \cD^N$. Taking the infimum over all $\Gamma \in \cD^N$, concavity of $E_\RDM$ is shown.

\noindent{\emph{Proof of 2.}} If $v_1\geq v_2$, then we automatically have
\[\Tr(\hat H(v_1)\Gamma)\geq \Tr(\hat H(v_2)\Gamma)\] for all $\Gamma\in \cD^N$, and hence $E_\RDM(v_1)\geq E_\RDM(v_2)$. 

\noindent{\emph{Proof of 3.}} Let $v\in \cR$ which by Remark \ref{Remark: relatively bounded sets} means that $\hat v\in \cR^N$. By Lemma~\ref{Lemma: R and R inf}, $H(v)$ is lower semibounded, which means that there exists a $c>0$ such that  
\[\Tr(\hat H(v)\Gamma)\geq -c\Tr(\Gamma)\]
for all $\Gamma\in (X^N)^+$. It follows by definition that $E_\RDM(v)>-\infty$.

\noindent{\emph{Proof of 4.}} It is a standard result that the infimum of a family of continuous functionals is upper semicontinuous. See e.g. \cite[Theorem 10.3]{RooijSchikhof1982}.

\noindent{\emph{Proof of 5.}}
One can check that $E_\RDM$ is bounded from below when restricted to the ball $B_{\|\cdot\|_{X^*}}(0,1/2)$. Indeed, if $v\in B_{\|\cdot\|_{X^*}}(0,1/2)$, then for any $\Gamma\in \cD^N$ and $\gamma := \Pi(\Gamma)$, we have
\[-\frac{1}{2}\left(\Tr(T\gamma)+\Tr(\gamma)\right)\leq v(\gamma) \quad \Longleftrightarrow\quad -\frac{1}{2}\Tr(\gamma)\leq \Tr((T\dot+v)\gamma).\] Since $\hat w$ is positive, we have 
\[\Tr(\hat  H(v)\Gamma)=\Tr((\hat T\dot + \hat w)\Gamma) + \hat v(\Gamma)\geq \Tr((T \dot + v)\gamma)\geq \frac{1}{2}\Tr(T\gamma) -\frac{1}{2}\Tr(\gamma)\geq -\frac{1}{2}\Tr(\gamma)=-\frac{N}{2}\]
for all $\Gamma\in (X^N)^+.$ Hence $E(v)\geq -N/2$. A classical result in convex analysis says that since $E_\RDM$ is a proper concave function, it follows that $E_\RDM$ is continuous on all of $\mathrm{int}(\mathrm{dom}(E_\RDM))$ \cite[5.20]{vanTielConvexAnalysis}. In particular, since $\cR\subset \mathrm{int}(\mathrm{dom}(E_\RDM))$, $E_\RDM$ is continuous on $\cR$. Another result then tells us then that $E_\RDM$ is locally Lipschitz in $\cR$ \cite[5.21]{vanTielConvexAnalysis}.
\end{proof}

We next introduce \emph{the reduced density matrix functional} \cite{GibneyBoynMazziotti2022}
\begin{equation}\label{Equation: F_RDM}
F_\RDM : X \to \R \cup \{+\infty\}, \quad F_\RDM(\gamma):=
    \inf_{\Gamma\mapsto \gamma} \hat w(\Gamma)
\end{equation}
where ``$\Gamma\mapsto \gamma$'', which is standard notation in the DFT literature, means that the infimum is taken over all $\Gamma\in \cD^N$ such that $\Pi(\Gamma)=\gamma$. If $\gamma\notin \cD$, we have $F_\RDM(\gamma)=+\infty$ by definition. 
\begin{lemma}\label{Lemma: FDM is convex}
    $F_\RDM$ is convex. Furthermore, we have 
    \[ |F_\RDM(\gamma)| \leq \|\hat{w}\|_{(X^N)^*} \|\gamma\|_X \text{ for all }\gamma\in X, \]
    i.e. $F_\RDM$ bounded with respect to the $X$-norm on $\cD$.
\end{lemma}
\begin{proof}
Convexity follows simply from the linearity of $\Gamma \mapsto \hat w( \Gamma)$ and $\Gamma\mapsto \Pi(\gamma)$:
    \begin{equation}
        \begin{split}
        F_\RDM (\lambda\gamma_1 &+(1-\lambda)\gamma_2) = \inf \left\{\hat w(\Gamma) \;\middle|\; \Gamma\in\cD^N,  \Gamma \mapsto \lambda\gamma_1 +(1-\lambda)\gamma_2\right\} \\
    &\leq \inf \left\{\hat w((\lambda \Gamma_1 + (1-\lambda)\Gamma_2)) \;\middle|\; \Gamma_1,\Gamma_2\in\cD^N, \Gamma_1 \mapsto \gamma_1 , \Gamma_2 \mapsto \gamma_2\right\} \\
    &=\inf \left\{ \lambda \hat w( \Gamma_1) | \Gamma_1\in \cD^N, \Gamma_1\mapsto \gamma_1\right \} +\inf \left\{ \lambda \hat w( \Gamma_2) | \Gamma_2\in \cD^N, \Gamma_2\mapsto \gamma_2\right\}\\
    &=\lambda F_\RDM (\gamma_1) +(1-\lambda)F_\RDM(\gamma_2).  
        \end{split}
    \end{equation}
The inequality above comes from the fact that the infimum on the second line is taken over a smaller set. 

Since $\hat w$ is bounded with respect to the $(X^N)^*$-norm, we have
\[- \|\hat w\|_{(X^N)^*}\|\gamma\|_X \leq - \|\hat w\|_{(X^N)^*}\|\Gamma\|_{X^N} \leq \hat w(\Gamma)\leq \|\hat w\|_{(X^N)^*}\|\Gamma\|_{X^N}\leq  \|\hat w\|_{X^*} \|\gamma\|_X \text{ for all }\gamma \in \cD, \]
where we used \eqref{Equation: Positive norm equality}. Taking the infimum over $\Gamma\mapsto \gamma$, we obtain the bound on $F_\RDM(\gamma)$.
\end{proof}
\begin{remark}
      Since $F_\RDM$ is locally bounded in $\cD$, it is in particular locally Lipschitz in the $X$-norm in $\ri(\cD)$ by Theorem 5.20 in \cite{vanTielConvexAnalysis}. However, if $\cH$ is infinite-dimensional, $\ri(\cD)$ can be shown to be empty, and hence this result does not help our analysis.
\end{remark}

\begin{lemma}\label{Lemma: Variational Principle}
    The following variational principle holds:
    \begin{equation}\label{Equation: Lieb-Gilbert variation principle}
    E_\RDM(v) = \inf_{\gamma \in \cD}\left( F_\RDM (\gamma) + \Tr(T\gamma) +  v(\gamma)\right). 
    \end{equation}
\end{lemma}
\begin{proof}
    By definition of $E_\RDM$, we have 
    \[E_\RDM(v) = \inf_{\Gamma\in D^N}\Tr(\hat H(v)\Gamma) = \inf_{\Gamma\in D^N}\left(\hat w(\Gamma) + \Tr\big(\hat T \Gamma\big) + \hat v(\Gamma)\right).\]
    By Lemma \ref{Lemma: second quantization expectation value}, we have $\Tr(\hat T\Gamma) + \hat v(\Gamma)=\Tr( T\gamma) + v(\gamma)$, where $\Gamma\mapsto \gamma$. Using the fact that $\Pi:\cD^N\to\cD$ is surjective by Theorem \ref{Theorem: Tr_{-1} is surjective}, we have
\[E_\RDM(v)=\inf_{\gamma\in \cD}\left(\inf_{\Gamma\to\gamma}\hat w(\Gamma)+\Tr\big( T\gamma\big) + v(\gamma)\right).\]
\end{proof}
The following theorem is adapted from Lieb's Theorem 4.4 in \cite{Lieb1983}. 
\begin{theorem}
Let $\{\gamma_n\}_{n\in\N}$ be a sequence in $\cD$ which is bounded in the $X$-norm, and suppose that $\gamma_n\xrightharpoonup{*}\gamma\in \cD$. Then, there exists a $\Gamma\in \cD^N$, such that $\Pi(\Gamma)=\gamma$ and
\begin{equation}\label{Equation: The best lim inf equation}
        \hat w( \Gamma)\leq \liminf_n F_\RDM(\gamma_n).\end{equation}\label{Theorem: The best lim inf theorem}
\end{theorem}

\begin{proof}
Since $F_\RDM$ is defined as an infimum over the preimage under $\Pi$, we can pick a $\Gamma_n \in \cD^N$ for every $n$ such that  $\Pi(\Gamma_n)=\gamma_n$ and 
\[F_\RDM(\gamma_n)\leq \hat w(\Gamma_n)\leq F_\RDM(\gamma_n)+ \frac{1}{n}.\]
Since $\|\Gamma_n\|_1=\Tr(\Gamma_n)=1$, the sequence $\{\Gamma_n\}_{n\in \N}$ is bounded in trace-norm. By the Banach-Alaoglu theorem, there is a subsequence converging to some $\Gamma\in \cS_1(\cH^N)$ in the weak-$*$ topology. Since $\hat w$ is positive, we have by equation \eqref{Equation: Positive norm equality} that $\|\Gamma_n\|_{X^N}\leq \|\gamma_n\|_X$, so $\{\Gamma_n\}_{n\in \N}$ is bounded in the $X^N$-norm. By Lemma \ref{Lemma: v is lower semicontinuous}, it follows that $\hat w(\Gamma)\leq \liminf_n \hat w(\Gamma_n)=\liminf F_\RDM(\gamma_n)$.  

Denote $\tilde\gamma := \Pi(\Gamma)$. For any $ K\in \cK(\cH)$, we have $\hat K\in \cK^+(\cH^N)$, and so 
\[\Tr(K\tilde\gamma) = \Tr\big(\hat K\Gamma\big)=\lim_{n\to \infty} \Tr\big(\hat K\Gamma_n\big) = \lim_{n\to \infty} \Tr(K\gamma_n) =\Tr(K\gamma).\] Since $K\in \cK(\cH)$ was arbitrary, we hence conclude that $\gamma=\tilde\gamma$.

It remains to show that $\Gamma\in \cD^N$. Since $|\Psi\rangle \langle \Psi|$ is a compact operator, we have $\langle \Psi,\Gamma\Psi\rangle =\Tr(|\Psi\rangle\langle \Psi|\Gamma) = \lim_{n\to \infty}\Tr(|\Psi\rangle \langle \Psi|\Gamma_n) = \lim_{n\to \infty}\langle \Psi,\Gamma_n\Psi\rangle \in [0,1]$, so $0\leq \Gamma\leq 1$. Furthermore, since $\Gamma \mapsto \gamma$, we have that $\Tr(\Gamma)=\Tr(\gamma)=1$. By equation \eqref{Equation: Positive norm equality}, we have $\|\Gamma\|_{X^N} \leq \|\gamma\|_X<+\infty$, and we conclude that $\Gamma\in \cD^N$.
\end{proof}

We obtain the following corollary, analogous to Corollary 4.5 in \cite{Lieb1983}.
\begin{corollary}\label{Corollary: FDM is lower semi-continuous.}
\leavevmode
     \begin{enumerate}
    \item $F_\RDM $ is norm and weak-$*$ lower semi-continuous with respect to the $\cS_1(\cH)$-topology. 
    \item Given $\gamma\in \cD$, there there exists a $\Gamma\mapsto \gamma$ such that $\hat w(\Gamma)=F_\RDM(\gamma)$.
\end{enumerate}
\end{corollary}
\begin{proof}[Proof of 1.]
Let $\gamma_n\xrightharpoonup{*} \gamma$. By Theorem \ref{Theorem: The best lim inf theorem}, there is a $\Gamma \mapsto \gamma$, and we have \[F_\RDM(\gamma) = \inf_{\Gamma'\mapsto \gamma}\hat w(\Gamma')  \leq \hat w( \Gamma)\leq \liminf_n F_\RDM(\gamma_n),\] which shows that $F_\RDM$ is weak-$*$ lower semi-continuous.

\noindent\emph{Proof of 2.} Let $\gamma_n = \gamma$. Then $\gamma_n\rightharpoonup\gamma$ trivially, and so by Theorem \ref{Theorem: The best lim inf theorem}, there is a $\Gamma\mapsto \gamma$ such that
\[F_\RDM(\gamma)\leq \hat w(\Gamma) \leq \lim \inf F_\RDM(\gamma_n)=F_\RDM(\gamma).\]

\end{proof}

The function $E_\RDM$ is seen to be a modified Legendre--Fenchel transform of $F_\RDM$. The latter is defined as the map $F_\RDM^* : X^* \to \R\cup\{+\infty\}$ with values
$$
        F_\RDM^*(v) := \sup_{\gamma\in X} \bigl( v(\gamma) - F_\RDM(\gamma) \bigr),
$$
and is, by virtue of being the supremum of a family of linear functionals, a convex lower semicontinuous function. It is readily seen that $E_\RDM(v) = -F_\RDM^*(-(T\dot{+} v))$. On the other hand, the biconjugate $F_\RDM^{**} = (F_\RDM^*)^* : X \to \R\cup\{+\infty\}$ is also convex lower semicontinuous, and The Fenchel--Moreau Theorem~\cite[Theorem 6.15]{vanTielConvexAnalysis} states that a convex lower semicontinuous function equals its own biconjugate if and only if the original function is convex lower semicontinuous. These observations lead to the following:
\begin{corollary}\label{Corollary: FRDM equals sup.}
    $F_\RDM : X \to\R \cup\{+\infty\}$ and $E_\RDM : X^* \to \R\cup\{-\infty\}$ are conjugate functions in the sense that
    \begin{align}
        E_\RDM(v) &= \inf_{\gamma\in X} \bigl( F_\RDM(\gamma) + \Tr((T\dot{+} v)\gamma) \bigr), \\
        F_\RDM(\gamma) &= \sup_{v \in X^*} \bigl( E_\RDM(v) - \Tr((T\dot{+} v)\gamma) \bigr). \label{Equation: F as sup}
    \end{align}
\end{corollary}
\begin{proof}
    We exhibit the right-hand side of \eqref{Equation: F as sup} as the biconjugate of $F_\RDM$. Since $E_\RDM(v) = -F_\RDM^*(-v)$, and using $T \in X^*$, we have for the right-hand side
    $$
        \sup_{v\in X^*} \bigl( -F_\RDM^*(-(T\dot{+} v)) - \Tr((T\dot{+} v)\gamma)\bigr) = \sup_{v'\in X^*} \bigl( v' (\gamma) - F^*_\RDM(v') \bigr) = F_\RDM^{**}(\gamma).
    $$
    where the substituion $v' = -(T\dot{+}v)$ was made.
\end{proof}

\section{Density Functional Theory}\label{Section: Relation to DFT}

In this section, we propose a general framework for DFT. We consider densities $\rho$ as integrable real-valued functions on some $\sigma$-finite measure space $\Omega$, which in most applied settings will be $\R^3$. We shall denote the set of measurable subsets of $\Omega$ by $\Sigma(\Omega)$, and the set of real-valued integrable functions on $\Omega$ by $L^1_\R(\Omega)$. The Hilbert space $\cH$ can be any separable Hilbert space, but the conventional choice for electrons with spin is $\cH=L^2(\R^3\times \Z_2)$. As in the previous section, $T$ is a positive operator on $\cH$, and $\hat w$ models the inter-particle interaction. We will relate reduced density matrices to densities by a construction which we will call the \emph{diagonal map}, denoted by $D(\gamma)=\rho$. For this construction, we shall need a projection-valued measure $P$ on $\Omega$. In the conventional setting, $P$ maps measurable subsets $S\subset \Omega$ to the operator on $L^2(\R^3\times \Z_2)$ given by the restriction to $S\times \Z_2$, i.e. $ [P(S) f](x,\sigma)=\mathbbm{1}_S(x)f(x,\sigma)$, where $\mathbbm{1}_S$ is the indicator function of $S$. However, this conventional setting is only one out of many possible use cases, such as DFT with spin-dependent densities (spin-DFT), or models with discrete densities. We frame many of our results in this general setting, but we will subsequently refer to the ``Conventional setting'' as the one described here. We make the distinction clear in the table below.  

\begin{center}
\begin{tabular}{ |c|c| } 
\hline
General setting & Conventional setting \\
\hline
 $\cH$ is a separable Hilbert space & $\cH = L^2(\R^3\times\Z_2)$ \\ 
 $\Omega$ is a $\sigma$-finite measure space & $\Omega = \R^3$ \\ 
$P:\Sigma(\Omega) \to \cB(\cH)^+$ is a projection-valued measure &  $[P(S) f](x,\sigma)=\mathbbm{1}_S(x)f(x,\sigma)$ \\
 $T$ is a positive operator on $\cH$ & $T= -\nabla^2$ \\ $w$ is positive and infinitesimally $T$-bounded & $\hat w = \sum_{i<j}|x_i-x_j|^{-1}$\\
\hline
\end{tabular}
\end{center}

By a careful construction of the diagonal map $D$, our mathematical framework for RDMFT will carry over to that of DFT in a natural way. In this framework, we will prove that the resulting universal density functional is lower semi-continuous and expectation-valued in the general setting. A classic result by Lieb \cite{Lieb1983} follows as a special case, namely that that the universal density functional is lower semi-continuous and expectation-valued.

\subsection{The diagonal map {\itshape D}}
We introduce the \emph{diagonal map} $D$ which relates $\gamma \in \cS^+(\cH)$ to $\rho \in L^1_\R(\Omega)$, i.e., $\rho = D(\gamma)$, and can be thought of as the predual mapping of $P$ in the sense that $D^*(\mathbbm{1}_S)=P(S)$ for all $S\in \Sigma(\Omega)$. In the conventional setting, the diagonal map is usually expressed as $\rho(x)=D\gamma(x):=\sum_{\sigma}\gamma(x,\sigma; x,\sigma)$. However, as explained by Lieb in \cite{Lieb1983}, the diagonal map cannot be defined pointwise formally, but needs to be defined in terms of some averaging over measurable neighborhoods.  We will obtain, that in the conventional setting 
\[D: \cS_1^\mathrm{sa}(L^2(\R^3\times \Z_2)) \to L^1_\R(\R^3),\] 
is the map that reproduces the usual density of an $N$-electron state.
When no confusion can arise we denote 
\[\|f\|_1 \equiv \|f\|_{L^1_\R(\Omega)}:=\int_\Omega |f|,\]
which is the same notation as used for the trace-norm on operators. For the definition of the diagonal map, we will use the notion of \emph{projection-valued measures}. We recall the definition in the following. For the purposes of this, we will denote the intrinsic measure on $\Omega$ by $\mu$, which we otherwise suppress for ease of notation.  

\begin{definition}
    Let $(\Omega,\mu)$ be a $\sigma$-finite measure space and let $\cH$ be a Hilbert space. A \emph{projection-valued measure on $\Omega$} is a map $P:\Sigma(\Omega,\mu)\to \cB(\cH)^+$ satisfying
    \begin{enumerate}
        \item $P(S)$ is an orthogonal projection for all $S\in\Sigma(\Omega,\mu)$,
        \item $P(\varnothing)=0$, and $P(\Omega)=1$,
        \item $P\left(\bigcup S_i\right)=\sum P(S_i)$ for a countable collection of disjoint sets $\{S_i\}_{i\in\N}$, 
        \item $P(S\cap S')=P(S)P(S')$ for all $S,S
    '\in \Sigma(\Omega,\mu)$.
    \end{enumerate}
    Furthermore, we say that $P$ is \emph{faithful} if $P(S)=0\Rightarrow \mu(S)=0$. 
\end{definition}
Notice that as a consequence of point 4 above, we have $P(S)P(S') = P(S')P(S)=0$ whenever $S$ and $S'$ are disjoint. This means that the subspaces corresponding to $P(S)$ and $P(S')$ are mutually orthogonal. 
With these notions in mind, we introduce the diagonal map $D$ by the following theorem: 

\begin{theorem}\label{Theorem: Diagonal map}
    Given a projection-valued measure $P:\Sigma(\Omega) \to \cB(\cH)^+$ and $\gamma\in \cS_1^\mathrm{sa}(\cH)$, let $D(\gamma)$ be the Radon--Nikodym derivative of the signed measure $S \mapsto \Tr(P(S)\gamma)$ on $\Omega$. Then, $D: \cS_1^\mathrm{sa}(\cH)\to L^1_\R(\Omega)$ is the unique linear map satisfying: 
\begin{enumerate}
    \item $\|D(\gamma)\|_{1}\leq \|\gamma\|_1 $ for all $\gamma\in \cS_1^\mathrm{sa}(\cH)$, 
    \item $\gamma\geq 0\Rightarrow D(\gamma)\geq 0 $ almost everywhere,
    \item $\Tr(\gamma) = \int_{\Omega} D(\gamma)$ for all $\gamma\in \cS_1^\mathrm{sa}(\cH)$,
    \item $\Tr(P(S)\gamma)=\int_S D(\gamma)$ for all $S\in \Sigma(\Omega)$
\end{enumerate}
\end{theorem}
\begin{proof}
    For every $\gamma\in\cS_1^\mathrm{sa}(\cH) $, the map $S \mapsto \Tr(P(S)\gamma)$ defines a signed measure on $\Omega$. By the Radon-Nykodim Theorem, there is a unique $D(\gamma)\in L^1_\R(\Omega,\mu)$ such that 
    \[\int_S D(\gamma)d\mu = \Tr(P(S)\gamma)\text{ for all }S\in \Sigma(\Omega,\mu). \]
Hence, the map $D$ is uniquely defined by point 4 of the theorem. 

We next show that $D$ is linear. Let $\gamma,\gamma'\in \cS_1^\mathrm{sa}(\cH)$. Then, by linearity of $\Tr$, we have for all $S\in \Sigma(\Omega,\mu)$ that  
\[\int_S D(\gamma+\gamma')d\mu=\Tr(P(S)(\gamma+\gamma'))=\Tr(P(S)\gamma)+\Tr(P(S)\gamma')=\int_S(D(\gamma)+D(\gamma'))d\mu, \] and hence $D(\gamma + \gamma')=D(\gamma)+D(\gamma')$. The equality $D(\lambda\gamma)=\lambda D(\gamma)$ also follows as a consequence of the linearity of $\Tr$ and $\int$. 

Point 3 of the theorem is a special case of point 4, and point 2 is a consequence of the unsigned Radon-Nykodim theorem. For point 1, observe that for the Jordan decomposition $\gamma^\pm = (|\gamma|\pm\gamma)/2$, we have 
\[\int_\Omega |D(\gamma)| = \int_\Omega|D(\gamma^+)-D(\gamma^-)|\leq \int_\Omega D(\gamma^+) +D(\gamma^-)=\Tr(\gamma^+ ) + \Tr(\gamma^-) = \|\gamma\|_1.\]
\end{proof}
We see that in the conventional setting, the above definition indeed reproduces the ``pointwise'' diagonal map. Indeed, if $\gamma(x,\sigma; y,\sigma')=\overline{\psi(x,\sigma)}\psi(y,\sigma')$ for some $\psi\in L^2(\R^3\times \Z_2)$, then point 4 implies in particular that 
\[\int_S \rho = \sum_\sigma \int_S |\psi(x,\sigma)|^2\]
for any measurable $S\subset \R^3$, implying that $\rho(x)=\sum_\sigma |\psi(x,\sigma)|^2$. Since $D$ is linear, we obtain for a general $\gamma$ with spectral decomposition $\gamma = \sum_n \lambda |\psi_n\rangle\langle\psi_n|$ that 
\[\rho(x) = \sum_{n,\sigma}\lambda_n|\psi_n(x,\sigma)|^2,\]
which is a more usual definition of the relation between $\gamma$ and $\rho$. 

\begin{remark}
    We claim that the construction in Theorem \ref{Theorem: Diagonal map} is in fact a one-to-one correspondence if $P$ is generalized to be a positive operator-valued measure: Given a linear map $D:\cS_1^\mathrm{sa}(\cH)\to L^1_\R(\Omega)$ satisfying points 1--3, there is a unique operator-valued measure $P$ satisfying point 4. The proof of this uses a generalized Radon-Nykodim theorem due to Pedersen and Takesaki \cite{PedersenTakesaki1973}. One can also generalize the domains of definition to $\cS_1(\cH)$ and $L^1(\Omega)$ if the complex Radon-Nykodim theorem is used. These results can also be generalized further, but we relegate details on these results to a future communication.   
\end{remark}

As $D : \cS_1^\mathrm{sa}(\cH)\to L^1_\R(\Omega)$ is bounded, its adjoint $D^* : L^\infty_\R(\Omega ) \to \cB^\mathrm{sa}(\cH)$ is also well-defined and bounded. The dual map $D^*$ is the natural map relating a potential $v$ defined on $\Omega$ to the corresponding operator $D^*(v)$ acting on $\cH$. By point 4 of Theorem \ref{Theorem: Diagonal map}, we have
\[\Tr(D^*(\mathbbm{1}_S)\gamma)=\int_\Omega \mathbbm{1}_SD(\gamma) = \int_SD(\gamma) = \Tr(P(S)\gamma)\]
for all $\gamma\in \cS_1^\mathrm{sa}(\cH)$ and all measurable $S\subset \Omega$, implying that $D^*(\mathbbm{1}_S)=P(S)$. As a consequence of this, we have that \begin{lemma}\label{Lemma: D* is a homomorphism of algebras}
    $D^*$ is a homomorphism of algebras, i.e., $D^*(v_1)D^*(v_2)=D^*(v_1v_2)$ for all $v_1,v_2\in L^\infty_\R(\cH)$. 
\end{lemma}

\begin{proof}
    It is enough to show this for simple functions, since they are dense and $D^*$ is bounded. Let $v_1,v_2\in L^\infty(\R)$ be simple functions, and let 
\[v_1 = \sum v_1^i \mathbbm{1}_{S_1^i}, \text{ and } v_2 = \sum v_2^i \mathbbm{1}_{S_2^i},\]
where $\{S_1^i\}_i$ and $\{S_2^i\}_i$ are mutually disjoint respectively. Since $D^*(\mathbbm 1_S)=P(S)$ for any measurable $S\subset \Omega$, we have for any two measurable sets $S,S'\subset \Omega$ that $D^*(\mathbbm{1}_S\mathbbm{1}_{S'}) = D^*(\mathbbm{1}_{S\cap S'}) = P(S\cap S')=P(S)P(S')=D^*(S)D^*(S')$. We hence get
\[D^*(v_1v_2)=\sum_{ij} v_1^iv_2^jD^*(\mathbbm{1}_{S_1^i}\mathbbm{1}_{S_2^j})=\sum_{ij}v_1^iv_2^jD^*(\mathbbm{1}_{S_1^i})D^*(\mathbbm{1}_{S_2^j})=D^*(v_1)D^*(v_2).\]
Hence, $D^*$ is a homomorphism of algebras. 
\end{proof}

As a consequence of the above lemma, we obtain:

\begin{lemma}\label{Lemma: Nice properties of D and D*}
    The diagonal map $D$ and its dual $D^*$ satisfy the following: 
    \begin{enumerate}
        \item $D^*$ is a positive map.
        \item $D^*$ respects the continuous functional calculus, i.e., $f(D^*(v))=D^*(f(v))$ for any continuous function $f$ and any $v\in L^\infty_\R(\cH)$.
        \item\label{Lemma: Nice properties point 3} $D(D^*(v)^{1/2}\gamma D^*(v)^{1/2})=vD(\gamma)$ for any $v\in L^\infty(\Omega)^+$ and $\gamma\in \cS_1^\mathrm{sa}(\cH)$.  
    \end{enumerate}
\end{lemma}

\begin{proof}
We first demonstrate point 2, that $D^*$ respects the continuous functional calculus. By Lemma \ref{Lemma: D* is a homomorphism of algebras}, we have 
\[p(D^*(v)) = D^*(p(v))\]
for any polynomial $p$ and any $v\in L^\infty_\R(\Omega)$. Using the fact that polynomials can arbitrarily approximate continuous functions in the $L^\infty$--norm, the result follows. 

Point 1 follows since if $v$ is positive, then by point 2, $|D^*(v)| = D^*(|v|)=D^*(v)$. Next, we demonstrate point 3. Let $S\subset \Omega$ be a measurable set. Since $D^*(v^{1/2}) \gamma D^*(v^{1/2}) \in \cS_1^\mathrm{sa}(\cH)$, we get
\begin{equation*}
\begin{split}
\int_S D(D^*(v^{1/2})\gamma D^*(v^{1/2}) &= \Tr(P(S)D^*(v^{1/2})\gamma D^*(v^{1/2})) =\Tr(D^*(\mathbbm{1}_S v) \gamma )
\\ &= \int_\Omega \mathbbm{1}_SvD(\gamma)=\int_S vD(\gamma).
\end{split}
\end{equation*}
We used the cyclicity of the trace and the algebra homomorphism in the second equality.
Since $S$ was arbitrary, we conclude that $D(D^*(v^{1/2})\gamma D^*(v^{1/2})) = vD(\gamma)$ for all $v\in L^\infty_\R(\Omega)^+$ and $\gamma\in \cS_1^\mathrm{sa}(\cH)$.
\end{proof}

\begin{theorem}\label{Theorem: Surjective diagonal map}
    Suppose that $P$ is a faithful projection-valued measure on $\Omega$, meaning that $S$ has zero measure if $P(S)=0$, then the corresponding diagonal map $D$ is surjective. Furthermore, $D$ maps $\cS_1^+(\cH)$ surjectively onto $L^1_\R(\Omega)^+$ and $D^*$ is an isometry.   
\end{theorem}

In order to prove Theorem \ref{Theorem: Surjective diagonal map}, we need a couple of technical lemmas: 

\begin{lemma}\label{Lemma: D maps positive guys surjectively onto positive guys}
    Let $\rho\in \ran(D)$, and suppose that $\rho \geq 0$. Then, there exists a $\gamma\in \cS_1(\cH)^+$ such that $D(\gamma)=\rho$. 
\end{lemma}
\begin{proof}
        As $\rho \in \ran(D)$, there exists a $\tilde\gamma\in\cS_1^\mathrm{sa}(\cH)$ such that $D(\tilde\gamma)=\rho$. Let $\tilde\gamma^+ := (|\tilde\gamma| + \tilde\gamma)/2$. As $\tilde\gamma^+\geq \tilde\gamma$, and since $D$ is a positive map we have $D(\tilde\gamma^+)\geq D(\tilde\gamma)$. Define
    \[v(x):= \begin{cases}
        \frac{D(\tilde\gamma)(x)}{D(\tilde\gamma^+)(x)} &\text{ if } x\in \supp D(\tilde\gamma), \\ 
        0 &\text{ else,}
    \end{cases}\]
    and define $\gamma= D^*(v)^{1/2}\tilde\gamma^+D^*(v)^{1/2}. $ By point \ref{Lemma: Nice properties point 3} of Lemma \ref{Lemma: Nice properties of D and D*}, we have 
    \[D(\gamma)=vD(\tilde\gamma^+)=D(\tilde\gamma)=\rho,\] which concludes the proof.
\end{proof}
\begin{lemma}\label{Lemma: Positive increasing sequence}
    Let $\rho$ be an element of the $L^1(\Omega)$-closure of $\ran(D)$, and suppose that $\rho\geq 0$. Then, there exists an increasing positive sequence $\{\rho_n\}\subset \ran(D)$, meaning $\rho_n\geq \rho_{n'}$ whenever $n\geq n'$, such that $\|\rho_n-\rho\|_1 \to 0$. 
\end{lemma}

\begin{proof}
Since $\rho$ is in the $L^1(\Omega)$-closure of $\ran(D)$, there exists a sequence $\tilde\rho_n$ in $\ran(D)$ such that $\|\tilde\rho_n-\rho\|_1 \to 0$. Let $\tilde\rho_n= \tilde\rho_n^+-\tilde\rho_n^-$ be a decomposition into positive parts with disjoint supports $S^+_n$ and $S^-_n$, and let $\gamma_n$ be a preimage of $\tilde\rho_n$ under $D$. Since $D(P(S^\pm_n)\gamma_n)=\mathbbm{1}_{S^\pm_n} D(\gamma_n) = \tilde\rho^\pm_n$, we have $\tilde\rho_n^\pm \in \ran (D)$. Furthermore, it is not hard to show that $|\rho-\tilde\rho_n^+|\leq |\rho-\tilde\rho_n|$, which implies that $\|\rho-\tilde\rho_n^+\|_1 \leq \|\rho-\tilde\rho_n\|_1 \to 0.$ Indeed, this follows from the following case-by-case consideration: If $x\in \Omega\setminus S_n^-$, then
\[\rho(x) - \tilde\rho_n^+(x)=\rho(x) - \tilde\rho_n(x).\]
If $x\in S^-_n$, then 
\[\rho(x) - \tilde\rho^+_n(x)=\rho(x)\leq \rho(x) + \tilde\rho_n^-(x) = \rho(x) - \tilde\rho_n(x).\]
Furthermore, since $\rho(x)\geq 0$, we have
\[\rho(x)-\tilde\rho^+_n(x)=\rho(x)\geq -\rho(x) \geq -\rho(x)-\tilde\rho_n^-(x) = -\rho(x)+\tilde\rho_n(x). \]
In conclusion, $|\rho-\tilde\rho_n^+|\leq |\rho-\tilde\rho_n|$.

We thus have a positive sequence $\{\tilde\rho_n^+\}_{n\in\N}$ converging to $\rho$. We restrict $\{\rho^+_n\}_{n\in\N}$ to a subsequence which converges pointwise to $\rho$ almost everywhere. Let $h_n := \sum_{k=1}^n \tilde\rho_k^+$, and define \[v_n(x):= \begin{cases}
    \frac{\rho(x)}{h_n(x)} &\text{ if } h_n(x) > \frac{1}{n}, \\ 0 & \text{ else. }
\end{cases}\]
As $h_n\in \ran(D)$, there exists a $\gamma_n\in \cB(\cH)^+$ such that $D(\gamma_n)=h_n$. By point \ref{Lemma: Nice properties point 3} of Lemma \ref{Lemma: Nice properties of D and D*}, $\tilde\gamma_n := D^*(v_n)^{1/2}\gamma_nD^*(v_n)^{1/2}$ satisfies $\rho_n:=D(\tilde\gamma_n)=v_nh_n = \rho \mathbbm{1}_{\{h_n>1/n\}} \in \ran(D)$. We next show that $\rho_n $ converges to $\rho$ pointwise almost everywhere. Consider the set
\[S := \bigcup_{n=1}^\infty \left\{h_n > \frac{1}{n}\right\}.\]
Let $x\in \Omega$. If $x\in S$, then $\mathbbm{1}_{\{h_n>1/n\}}(x)=1$ for all $n$ large enough, so $\rho_n(x)\to\rho(x)$. If $x\in \{\rho = 0\} \setminus S$, then $\rho_n(x)=0$ for all $n$ and hence $\rho_n(x)=0=\rho(x)$. Now, suppose that $x\in \supp(\rho)\setminus S$. Then, $\sum_{k=1}^n\tilde\rho_n^+(x)=h_n(x)\leq  1/n$ for all $n$, and so $\tilde\rho_n^+(x)=0$ for all $n$. Hence, \[\supp(\rho)\setminus S \subset \{x \in \Omega \;|\; \tilde\rho_n^+(x)\not\to \rho(x)\}.\] Since $\{\tilde \rho_n^+\}_{n\in\N}$ converges to $\rho$ pointwise almost everywhere, the set $ \{x \in \Omega \;|\; \tilde\rho_n^+(x)\not\to \rho(x)\}$ has measure zero, and thus, we have that $\rho_n \to \rho$ pointwise almost everywhere. By the dominated convergence theorem, we conclude that $\|\rho_n - \rho\|_1 \to 0$. 
\end{proof}
\begin{lemma}\label{Lemma: Dense range}
    If $P$ is faithful, meaning that $S$ has measure zero if $P(S)=0$, then $D^*$ is an isometry. In particular, $D$ has dense range. 
\end{lemma}
\begin{proof}

Let $v\in L^\infty_\R (\Omega)$. Let $\{v_n\}_{n\in\N}$ be a sequence of simple functions converging to $v$. For every $n$, let \[v_n = \sum v^+_{n,i}\mathbbm{1}_{S_{n,i}^+} - \sum v_{n,i}^-\mathbbm{1}_{S_{n,i}^-}\] denote the decomposition of $v_n$ into constant parts, where the sets $\{S^\pm_{n,i}\}_i$ are mutually disjoint. Using the fact that $D^*(\mathbbm{1}_{S})=P(S)$ for any measurable set $S$, we have 
\[D^*(v_n)=\sum v^+_{n,i}P(S_{n,i}^+) - \sum v_{n,i}^-P(S_{n,i}^-).\]
Since the sets $\{S^\pm_{n,i}\}_i$ are mutually disjoint, the projections $P(S_{n,i}^\pm)$ are mutually orthogonal, and hence, 
\[\|D^*(v_n)\| = \max_i v_{n,i}^\pm  = \|v_n\|_\infty.\]
By continuity of $D^*$, we have \[\|D^*(v)\| = \|\lim_n D^*(v_n)\| = \lim \|D^*(v_n)\| = \lim_n\|v_n\|_\infty = \|v\|_\infty ,\] implying that $D^*$ is an isometry. In particular, $D^*$ is injective, and by a classical result in functional analysis (see e.g. \cite[Theorem 4.12]{Rudin1991}), we conclude that $D$ has dense range. 

\end{proof}
With the preceding lemmas in mind, we are now ready to prove Theorem \ref{Theorem: Surjective diagonal map}. 
\begin{proof}[Proof of Theorem \ref{Theorem: Surjective diagonal map}] Let $\rho\in L^1(\Omega)^+$. Since $P$ is faithful, $\ran D$ is dense in $L^1(\Omega)$ by Lemma \ref{Lemma: Dense range}. By Lemma \ref{Lemma: Positive increasing sequence}, there exists a positive increasing sequence $\{\rho_n\}_{n\in\N} \subset \ran D$ converging to $\rho$. Since $\{\rho_n\}_{n\in\N}$ is an increasing sequence, $\rho_{n+1}-\rho_n\geq 0$, and as $\rho_{n+1}-\rho_n\in \ran(D)$, there exists by Lemma \ref{Lemma: D maps positive guys surjectively onto positive guys} a $\gamma_n\in \cB(\cH)^+$ with $D(\gamma_n)=\rho_{n+1}-\rho_n$ for every $n$. Moreover, since $D$ is a positive map, 
\[\|\gamma_n\|_1 = \|D(\gamma_n)\|_1 =\|\rho_{n+1}-\rho_n\|_1.\] Next, set $\tilde\gamma_m := \sum_{n=1}^m \gamma_n$. Then, if $m\geq n$, we have
\[\|\tilde\gamma_m-\tilde\gamma_n\|_1 = \left\|\sum_{k=n+1}^m\gamma_k\right\|_1 = \left\|\sum_{k=n+1}^mD(\gamma_k)\right\|_1 = \|\rho_m-\rho_n\|_1.\]
Hence, $\{\tilde\gamma_n\}_{n\in\N}$ is a Cauchy sequence in $\cS_1(\cH)$. Denote $\tilde\gamma:= \lim\tilde\gamma_n$. Then, by continuity of $D$, we have
\[D(\tilde\gamma)=D(\lim\tilde\gamma_n)=\lim D(\tilde\gamma_n) = \lim \rho_{n+1} = \rho.\]

Hence, we have shown that there exists a $\tilde\gamma\in \cB(\cH)^+$ with $D(\tilde\gamma)=\rho$. Next, if $\rho\in L^1(\Omega)$, we can decompose $\rho$ into positive and negative parts, $\rho= \rho^+ -\rho^-$, and by the first part of the proof, there exist $\tilde\gamma^\pm\in \cB(\cH)^+$ with $D(\tilde\gamma^\pm)=\rho^\pm$. Hence $D(\tilde\gamma^+ -\tilde\gamma^-)=\rho$, which concludes the proof. 
    
\end{proof}

\begin{remark}\label{Remark: D maps B+ surjectively ont L1+}
    Note that since $\ran D=L^1_\R(\Omega)$, Lemma \ref{Lemma: D maps positive guys surjectively onto positive guys} says that $D$ maps $\cS_1^+(\cH)$ surjectively onto $L^1(\Omega)^+$. 
\end{remark}

\subsection{General setting for DFT}
We restrict the domain of the diagonal map to $X$. Define
\[\Xi := D(X),\]
and endow $\Xi$ with the norm
\begin{equation}\label{Equation: Xi norm definiton}
    \|\rho\|_\Xi := \inf\left\{\|\gamma\|_X \;\middle|\; \gamma\in X,\; D\gamma = \rho\right\}. 
\end{equation}
Note that the definition above ensures that 
\begin{equation}\label{Equation: Xi norm inequality}
    \|D(\gamma)\|_\Xi\leq \|\gamma\|_X \text{ for all } \gamma\in X .
\end{equation} 
The space $\Xi$ should be thought of as the space of \emph{finite-kinetic energy densities} $\rho$. 

\begin{theorem} \label{Theorem: Xi is a Banach space}
 $\|\cdot\|_{\Xi}$ defines a norm on $\Xi$, and $\Xi$ is complete with respect to this norm, i.e. $\Xi$ is a real Banach space.
 \end{theorem}

\begin{proof}
Firstly, let us check that $\|\cdot\|_{\Xi}$ is a norm. If $\lambda\in \R$, then it follows from linearity of $D$ and $\|\lambda\gamma\|_X=|\lambda|\|\gamma\|_X$ that $\|\lambda\rho\|_\Xi=|\lambda|\|\rho\|_\Xi.$ If $\|\rho\|_\Xi=0$, there exists a sequence $\{\gamma_n\}_{n\in\N}$ such that $D\gamma_n=\rho$ and $\|\gamma_n\|_X\to 0$. In particular, by point (1) of Theorem \ref{Theorem: Diagonal map} and equation \eqref{Equation: Xi norm inequality}, we have $\|\rho\|_1 =\|D(\gamma_n)\|_1\leq \|\gamma_n\|_1\leq \|\gamma_n\|_X\to 0$, and hence $\rho=0$. Next, we check that the triangle inequality holds. Let $\rho,\rho'\in \Xi$. Then, we have
\begin{align*}
    \|\rho+\rho'\|_\Xi &= \inf\{\|\tilde\gamma\| \;|\; D(\tilde\gamma) = \rho+\rho'\} \\
    &\leq \inf \{ \|\gamma + \gamma'\| 
    \;|\; D(\gamma)= \rho,\; D(\gamma')=\rho'\} \\
    &\leq \inf \{\|\gamma\|_X+\|\gamma'\|_X  \;|\; D(\gamma)=\rho, \; D(\gamma')=\rho'\} \\
    &=\|\rho\|_\Xi+\|\rho'\|_\Xi,
\end{align*}
where the first inequality comes from the fact that the infimum is taken over a smaller set, and the second inequality comes from the triangle inequality of the $X$-norm. It remains to show that $\Xi$ is complete. Since $D$ is surjective onto its image, the first isomorphism theorem says that
\begin{align*}
    \tilde D: X/{\ker D} & \to \Xi, \\
    [\gamma] & \mapsto D(\gamma)
\end{align*}
is an isomorphism of vector spaces. Furthermore, since $D$ is bounded, $\ker D$ is closed, which makes $X/{\ker D}$ a Banach space when equipped with the quotient norm \cite[Theorem 1.41]{Rudin1991}
\[\left\|[\gamma]\right\|_{X/{\ker D}} = \inf_{\gamma'\in \ker D} \|\gamma + \gamma'\|_X.\]
It is clear that the isomorphism between $X/{\ker D}$ and $\Xi$ is an isometry: 
\[\|D(\gamma)\|_\Xi = \inf\{\|\gamma'\|_X \; |\;\gamma'\in X, \; D(\gamma')=D(\gamma)\} = \inf \{\|\gamma' + \tilde\gamma\|_X \; | \; \tilde\gamma\in \ker D\} = \left\|[\gamma]\right\|_{X/{\ker D}}.\]
Hence, $\Xi$ is isometrically isomorphic to the Banach space $X/{\ker D}$.
\end{proof}
Denote by $\Xi^+$ the set of positive functions in $\Xi$. 

\begin{lemma}\label{Lemma: Xi norm equals infimum over decompositions.}
    The $\Xi$-norm satisfies 
    \begin{align*}
        \|\rho\|_\Xi & = \inf_{\rho^\pm\geq 0, \;\rho=\rho^+-\rho^-}\left(\|\rho^+\|_\Xi + \|\rho^-\|_\Xi\right) \\ &= \inf_{\rho^\pm\geq0,\; \rho=\rho^+-\rho^-}\left(\inf_{\gamma^+\in X^+, \;D(\gamma^+)=\rho^+}\Tr((1+T)\gamma^+) + \inf_{\gamma^-\in X^+, \;D(\gamma^-)=\rho^-}\Tr((1+T)\gamma^-)\right).
    \end{align*}
\end{lemma}
\begin{proof}
Observe that by definition, 
    \[\|\rho\|_\Xi = \inf_{\gamma^\pm\in X^+,\;\rho=D(\gamma^+)-D(\gamma^-)}\left(\Tr((1+T)\gamma^+) +\Tr((1+T)\gamma^-)\right).\]
    Since $\rho = D(\gamma^+) - D(\gamma^-)$ is a decomposition of $\rho$ into positive and negative parts, we have in particular that
    \[\|\rho\|_\Xi\geq \inf_{\rho^\pm\geq0, \;\rho=\rho^+-\rho^-}\left(\inf_{\gamma^+\in X^+,\; D(\gamma^+)=\rho^+}\Tr((1+T)\gamma^+) + \inf_{\gamma^-\in X^+, \;D(\gamma^-)=\rho^-}\Tr((1+T)\gamma^-)\right).\]
    Since $\Tr((1+T)\gamma^\pm)=\|\gamma^\pm\|_X$, for positive $\gamma^\pm$, we get by expanding to non-negative $\gamma^\pm$,
    \begin{align*}
        \|\rho\|_\Xi \geq& \inf_{\rho^\pm\geq0, \;\rho=\rho^+-\rho^-}\left(\inf_{\gamma^+\in X^+,\; D(\gamma^+)=\rho^+}\Tr((1+T)\gamma^+) + \inf_{\gamma^-\in X^+, \;D(\gamma^-)=\rho^-}\Tr((1+T)\gamma^-)\right)\\
        \geq&\inf_{\rho^\pm\geq0, \;\rho=\rho^+-\rho^-}\left(\inf_{\gamma^+ \in X,\; D(\gamma^+)=\rho^+}\|\gamma^+\|_X + \inf_{\gamma^-\in X,\; D(\gamma^-)=\rho^-}\|\gamma^-\|_X\right) \\
        =&\inf_{\rho^\pm\geq 0,\; \rho=\rho^+-\rho^-}\left(\|\rho^+\|_\Xi + \|\rho^-\|_\Xi\right) \geq \|\rho\|_\Xi,
    \end{align*}
where the last inequality is simply an application of the triangle inequality. 
\end{proof}

\begin{remark}\label{Remark: Positive norm conjecture}
   In light of the second equality in Lemma \ref{Lemma: Xi norm equals infimum over decompositions.}, we conjecture that  
    \begin{equation}\label{Equation: Positive norm supremum}
        \|\rho\|_\Xi = \inf\{\Tr((1+T)\gamma) \;| \;\gamma\geq 0, \; D(\gamma)=\rho\} \text{ for all }\rho\in \Xi^+,
    \end{equation}
    although we are not able to provide a proof for this. 
\end{remark}

Next, we introduce the universal density functional $F$. It can be defined in terms of the reduced density matrix functional by 
\[F(\rho) : =\inf \{F_\RDM(\gamma) + \Tr(T\gamma) \;|\; \gamma\in \cD,\; D\gamma=\rho\}.\]
We see immediately from the definition that $F_{D}(\rho)=+\infty$ iff $\rho\notin D(\cD)$, i.e. $D(\cD)$ is the effective domain of $F.$ This domain is optimal in the following sense: $D(\cD)$ consists of precisely all densities which correspond to at least one density matrix of finite kinetic energy.

The following theorem is adapted from Lieb's Theorem 4.4 in \cite{Lieb1983}: 

\begin{theorem}\label{Theorem: The best lim inf theorem for density}
    Let $\{\rho_n\}_{n\in\N}$ be a bounded sequence in $D(\cD)$ such that $\rho_n\xrightharpoonup{L^1_\R(\Omega)} \rho\in \Xi^+$. Then, there exists a $\Gamma\in \cD^N$ such that $(D\circ\Pi)(\Gamma) = \rho$ and we have
    \begin{equation}
    \Tr(\hat H_{\hat w}\Gamma)\leq \lim \inf F(\rho_n).\end{equation}
\end{theorem}
\begin{proof}
Notice that the statement holds trivially if $\liminf F(\rho_n)=+\infty$. Thus, we suppose that $\liminf F(\rho_n)<+\infty$.  Similarly as in the proof of Theorem \ref{Theorem: The best lim inf theorem}, we pick for every $n$ a $\Gamma_n \in \cD^N$ such that $\Pi(\Gamma_n)=\rho_n$ and 
    \[F(\rho_n) \leq \Tr(\hat{H}_{\hat w}\Gamma_n) \leq F(\rho_n) + \frac{1}{n},\]
where $\hat{H}_{\hat w} := \hat{T}\dot{+}\hat{w}$. Since $\|\Gamma_n\|_1 = \Tr(\Gamma_n)=1$, we can extract a weak-$*$ convergent subsequence by the Banach-Alaoglu theorem, converging to some $\Gamma\in \cS_1(\cH^N)$. Since $\hat H_{\hat w}$ is positive, we have by Lemma \ref{Lemma: Trace of unbounded operator is lower semicontinuous} and \ref{Lemma: Trace inequality} that
\begin{equation*}
\|\Gamma\|_{X^N} = \Tr\big((1+\hat{T})\Gamma\big) \leq 1 + \Tr(\hat H_{\hat w}\Gamma) \leq 1 + \liminf_{n\to\infty} \Tr(\hat H_{\hat w}\Gamma_n) = 1 + \liminf_{n\to\infty} F(\rho_n)<\infty,
\end{equation*}
such that $\Gamma\in X^N$. we next show that $\Gamma\mapsto \rho$. Let $\tilde\rho := (D\circ\Pi)(\Gamma)$, let $v\in L^\infty(\Omega)^+$ and denote $V:=D^*(v)\in \cB^+(\cH)$. We have by Lemma \ref{Lemma: Trace of unbounded operator is lower semicontinuous} that 
\[\int_\Omega v \tilde\rho=\Tr(\hat V \Gamma) \leq \liminf_{n\to \infty} \Tr(\hat V \Gamma_n) = \lim_{n\to \infty} \int_\Omega v\rho_n = \int_\Omega v\rho \]
    This means that for any $v\in L^\infty(\Omega)^+$ we have
    \[\int_\Omega v(\rho-\tilde\rho)\geq 0\]
    In particular, this implies that $\rho-\tilde\rho \geq 0.$ But $\int \rho-\tilde\rho=N-N=0$, so $\rho=\tilde\rho$. By Theorem \ref{Theorem: Diagonal map}, we have $\Tr(\Gamma)=\Tr(\Pi(\Gamma)) = \int \rho = N$, and Since $|\Psi\rangle \langle \Psi|$ is a compact operator, we have $\langle \Psi,\Gamma\Psi\rangle =\Tr(|\Psi\rangle\langle \Psi|\Gamma) = \lim_{n\to \infty}\Tr(|\Psi\rangle \langle \Psi|\Gamma_n) = \lim_{n\to \infty}\langle \Psi,\Gamma_n\Psi\rangle \in [0,1]$, so $0\leq \Gamma\leq 1$. Hence $\Gamma\in \cD^N$. 
\end{proof}

\begin{remark}
The preceding proof is general and does not need the Rellich--Kondrachov theorem used in the proof of \cite[Theorem 4.4]{Lieb1983}. The original proof by Lieb is only valid for the special case of the conventional setting $\Omega=\R^3$ and $T=-\nabla^2$. 
\end{remark}

 We infer the following corollary from Theorem \ref{Theorem: The best lim inf theorem for density} in the same way as Corollary \ref{Corollary: FDM is lower semi-continuous.} follows from Theorem \ref{Theorem: The best lim inf theorem}.

\begin{corollary}\label{F is lower semicontinuous}
\leavevmode
\begin{enumerate}
    \item $F$ is lower semicontinuous with respect to the weak topology on $L^1_\R(\Omega)$. 
    \item Given $\rho\in D(\cD)$, there exits a $\Gamma\in X^N$ such that $(D\circ \Pi)(\Gamma)=\rho$ and $\Tr(\hat{H}_{\hat w}\Gamma)=F(\rho)$. 
\end{enumerate}
\end{corollary}
\begin{proof}[Proof of (1):] Let $\rho_n\rightharpoonup \rho$. By Theorem \ref{Theorem: The best lim inf theorem for density}, there is a $\Gamma \mapsto \rho$, and we have \[F(\rho) = \inf_{\Gamma'\mapsto \rho}\Tr(\hat H_{\hat w}\Gamma')  \leq\Tr(\hat H_{\hat w} \Gamma)\leq \lim \inf F(\rho_n),\] which shows that $F$ is weakly lower semi-continuous. 
\textit{Proof of (2):} Let $\rho_n = \rho$. Then $\rho_n\rightharpoonup\rho$ trivially, and so by Theorem \ref{Theorem: The best lim inf theorem for density}, there is a $\Gamma\mapsto \rho$ such that
\[F(\rho)\leq\Tr(\hat H_{\hat w}\Gamma) \leq \lim \inf F(\rho_n)=F(\rho).\]
\end{proof}
Since the $\Xi$-norm is stronger than the $L^1_\R(\Omega)$-norm, it follows that $F$ is also lower semicontinuous with respect to the weak $\Xi$-topology and the $\Xi$-norm. As a consequence, we have that
\begin{theorem}\label{Theorem: Lots of cool properties of FD}
The density functional satisfies the following:
\begin{enumerate}
    \item $F$ is convex,
    \item $F(\rho)$ is finite for all $\rho\in D(\cD)$, 
    \item $F$ is locally Lipschitz in $\ri D(\cD)$,  
    \item $F$ is subdifferentiable in $\ri D(\cD)$. 
\end{enumerate}
\end{theorem}
\begin{proof}
    Point 1: Convexity follows from the definition of $F$ in the same way as for $F_\RDM$ in Lemma \ref{Lemma: FDM is convex}. {Point 2} follows by definition, since $\dom(F)=D(\cD)$. By 5.21 in \cite{vanTielConvexAnalysis}, Point 3 follows from Point 1 if we can show that $f$ is locally bounded at some point in $\ri D(\cD)$. Since $D$ is surjective and $F_\RDM$ is locally bounded in $\ri(\cD)$, it follows that $F$ is locally bounded in $\ri(\cD)$. Point 4 is a consequence of 5.35 in \cite{vanTielConvexAnalysis}.  
\end{proof}
\begin{remark}
    The set $\ri D(\cD)$ might be empty. It turns out that $\ri D(\cD)$ is non-empty in certain special cases, for instance if $\Omega $ is one-dimensional and compact \cite{sutterSolutionVrepresentabilityProblem2024}.
\end{remark}

For $v\in \Xi^*$, the energy can be expressed as
\[E(v) = \inf_{\rho \in D(\cD)}\left[F(\rho) + \int v\rho\right]\]
Since
\[\int v\rho = \Tr(D^*(v)\gamma)\]
for all $\gamma \mapsto \rho$, we can also write
\[E(v)=\inf_{\gamma\in \cD}\left[ F_\RDM(\gamma) + \Tr((T+D^*(v))\gamma)\right] = E_\RDM(D^*(v)).\]
\begin{corollary}\label{Corollary: F equals sup.}
    $F : \Xi \to\R \cup\{+\infty\}$ and $E : \Xi^* \to \R\cup\{-\infty\}$ are conjugate functions in the sense that
    \begin{align}
        E(v) &= \inf_{\rho\in \Xi} \left( F(\rho) + \int_\Omega v\rho \right), \\
        F(\rho) &= \sup_{v \in \Xi^*} \left( E(v) - \int_\Omega v\rho \right). \label{Equation: F density as sup}
    \end{align}
\end{corollary}
\begin{proof}
    This is essentially the same as Corollary \ref{Corollary: FRDM equals sup.}. Cf. \cite[equations (3.17) and (3.20)]{Lieb1983}.
\end{proof}

\subsection{Conventional setting for DFT}
Recall the conventional setting for DFT described in the beginning of section \ref{Section: Relation to DFT}. In this setting, we have $\cH=L^2(\R^3\times \Z_2)$, $\Omega =\R^3$ and $T=-\nabla^2$, and $P$ mapping measurable $S\subset \R^3$ to the corresponding indicator function $\mathbbm{1}_S$. The dual $D^*$ of $D$ maps a potential $v\in \Xi^*$ to the corresponding possibly unbounded operator on $\cH$ given by pointwise multiplication.

\begin{lemma}\label{Lemma: D maps X+ surjectively onto Xi+}
$D$ maps $\cD$ surjectively onto $\{\rho\in \Xi^+ | \int \rho = N\}$, i.e.
\[ D(\cD) = \left\{\rho\in \Xi^+ \middle| \int \rho = N\right\}.\]
\end{lemma}
\begin{proof}
By Remark \ref{Remark: D maps B+ surjectively ont L1+}, we know that $D(X^+)=\Xi^+$, and so it is clear that ${D(\cD)\subset\{\rho\in \Xi^+|\int \rho=N\}}$. 
For the opposite inclusion, suppose that $\rho \in \Xi^+$ satisfies $\int \rho=N$. By \cite[Theorem 1.2]{Lieb1983}, there exists a $\Psi\mapsto \rho$ with $|\Psi\rangle\langle \Psi|\in \cD^N$. In particular, $\Pi(|\Psi\rangle \langle \Psi|)$ is an element of $\cD$ and maps to $\rho$ under $D$. 
\end{proof}

\begin{lemma}\label{Lemma: Xi-inequalities}
   The $\Xi$-norm satisfies    \begin{equation}\label{Equation: norm inequality}
 \|\rho\|_1 + 3(\pi/2)^{4/3}\|\rho\|_3 \leq\|\rho\|_\Xi\leq   \|\rho\|_1 + (4\pi)^2N\left\|\nabla\sqrt{|\rho|}\right\|_2^2
   \end{equation} 
for all $\rho\in\Xi$. 
\end{lemma}
\begin{proof}
We first prove the second inequality in \eqref{Equation: norm inequality}. Suppose first that $\rho\in \Xi^+$, and let $\tilde\rho= N \rho/\|\rho\|_1$. Following Lieb's proof of Theorem 1.2 in \cite{Lieb1983}, there exists a $\Psi\in Q(\hat T)$ satisfying $\|\Psi\|_1^2=N$ such that $|\Psi\rangle \langle \Psi|$ maps to $\tilde \rho$, and $\|\hat T^{1/2}\Psi\|^2\leq (4\pi)^2N^2\left\|\nabla \sqrt{\tilde \rho}\right\|_2^2$. We thus have 
\[\frac{N}{\|\rho\|_1}\|\rho\|_\Xi = \|\tilde \rho\|_\Xi\leq \inf \{\Tr((1+\hat T)\Gamma) | \Gamma \mapsto \rho, \Gamma\geq 0\} \leq \|\Psi\|^2 + \|\hat T^{1/2}\Psi\|^2\] \[\leq 1 +(4\pi)^2N^2\|\nabla \sqrt{\tilde\rho}\|_2^2 = N +\frac{(4\pi)^2N^2}{\|\rho\|_1}\|\nabla \sqrt{\rho}\|_2^2.\]
By multiplying with $\|\rho\|_1/N$, we get
\[\|\rho\|_\Xi \leq \|\rho\|_1 +(4\pi)^2N\|\nabla \sqrt{\rho}\|_2^2.\]

Next, let $\rho\in \Xi$ be a general density, and decompose $\rho = \rho^+ - \rho^-$, where $\rho^+$ and $\rho^-$ are positive and have disjoint domains. By the triangle inequality, we get 
\[\|\rho \|_\Xi = \|\rho^+-\rho^-\|_\Xi\leq \|\rho^+\|_1 +\|\rho^-\|_1 +(4\pi)^2N\left(\|\nabla\sqrt{\rho^+}\|_2^2+\|\nabla\sqrt{\rho^-}\|_2^2\right)\]
\[=\|\rho\|_1 + (4\pi)^2N\left\|\nabla\sqrt{|\rho|}\right\|_2^2. \] The last equality follows from the fact that $|\rho| = \rho^+ + \rho^-$ and that $\rho^+$ and $\rho^-$ have disjoint domains. 

Next, we show the first inequality in \eqref{Equation: norm inequality}. For every $\gamma\in X$ such that $D\gamma=\rho$, pick a decomposition $\gamma= \gamma^+ - \gamma^-$, and denote 
\[\gamma^+ = \sum_n \lambda_n^+|\psi^+_n\rangle\langle \psi^+_n|, \text{ and } \gamma^- = \sum_n\lambda_n^- |\psi^-_n\rangle\langle \psi^-_n|\] their spectral decompositions. By the Sobolev inequality \cite[Eq. (1.10)]{Lieb1983} , we have 
\[  \|T^{1/2}\psi_n^\pm\|^2_2\geq 3(\pi/2)^{4/3}\|\rho_n^\pm\|_3,\]
for every $n$, where $\rho_n(x) = D(|\psi_n^\pm\rangle\langle \psi^\pm_n|)(x) = \sum_\sigma |\psi_n^\pm|^2(x,\sigma)$. Using the fact that \[ \rho =  \sum_n\lambda_n^+\rho_n^+-\lambda_n^-\rho_n^-,\] and $\|\rho_n^\pm\|_1=1$, it follows that
\begin{align*}
    \Tr((1+T)\gamma^+ ) + \Tr((1+T)\gamma^-) &= \sum_n \lambda_n^+ (1+\|T^{1/2}\psi^+_n\|^2) + \lambda_n^- (1+\|T^{1/2}\psi^-_n\|^2) \\
    &\geq \sum_n \lambda_n^+ +\lambda_n^- + 3(\pi/2)^{4/3}\sum_n\|\lambda_n^+\rho_n^+\|_3+\|\lambda_n^-\rho_n^-\|_3\\
    &\geq \bigg\| \sum_n \lambda_n^+\rho_n^+ -\lambda_n^-\rho_n^-\bigg\|_1 +3(\pi/2)^{4/3}\bigg\|\sum_n\lambda_n^+\rho_n^+-\lambda_n^-\rho_n^-\bigg\|_3 \\
    & = \|\rho\|_1 +3(\pi/2)^{4/3}\|\rho\|_3,
\end{align*} 
and hence 
\[\Tr((1+T)\gamma^+ ) + \Tr((1+T)\gamma^-) \geq \|\rho\|_1 + 3(\pi/2)^{4/3}\|\rho\|_3.\] Taking the infimum over decompositions $\gamma = \gamma^+ - \gamma^-$, we get the desired inequality. 
\end{proof}

\begin{lemma} \label{Lemma: First Xi inequality}
    Let $\rho\in \Xi^+$. Then, we have 
    \[\|\rho\|_1 + \left\|\nabla \sqrt{\rho}\right\|_2^2\leq \Tr((1+T)\gamma)\] for all $\gamma\in X^+$ satisfying $D(\gamma)=\rho$. 
\end{lemma}
\begin{proof}
By Theorem 1.1 in \cite{Lieb1983}, we have
$\left\|\nabla\sqrt{\rho}\right\|_2^2\leq \|\hat T^{1/2}\Psi\|^2 = \Tr(\hat T|\Psi\rangle\langle \Psi|)$ for any $\Psi$ that maps to $\rho$. Since any positive $\Gamma$ is a convex combination of pure states, we get
\[\|\rho\|_1 + \left\|\nabla\sqrt{\rho}\right\|_2^2\leq \Tr((1+\hat T)\Gamma)= \|\Gamma\|_{X^N}\] for any positive $\Gamma$ which maps to $\rho$. Since $\Pi : X^N\to X$ is positive and surjective, we get that 
\[\|\rho\|_1 +\left\|\nabla\sqrt{\rho}\right\|_2^2\leq \Tr((1+ T)\gamma) \] for any positive $\gamma\mapsto \rho$. 
\end{proof}

Next, we introduce the norm 
\begin{equation}
    \begin{split}
    \|\rho\|_{\Xi}':&= \inf \left\{\int(\rho^++\rho^-)+\int (\nabla \sqrt{\rho^+})^2 +\int (\nabla \sqrt{\rho^-})^2 \;\middle|\; \rho = \rho^+ - \rho^-, \quad \rho^\pm\in\Xi^+\right\} \\
& = \inf \left\{ \|\sqrt{\rho_+}\|_{H^1}^2 + \|\sqrt{\rho_-}\|_{H^1}^2 \;\middle|\; \rho = \rho^+ - \rho^-, \quad \rho^\pm\in\Xi^+\right\},
    \end{split}
\end{equation}
where $\|\cdot\|_{H^1}$ denotes the Sobolev norm $\|f\|_{H^1}^2 := \int f^2 + \int (\nabla f)^2.$ Lemma \ref{Lemma: Xi-inequalities} then implies that 
\begin{corollary}\label{Corollary: Equivalent norms}
    $\|\cdot\|_{\Xi}'$ defines a norm on $\Xi$ which is equivalent to $\|\cdot\|_{\Xi}$.
\end{corollary}
\begin{proof} 
We first show that $\|\cdot\|_\Xi'$ is indeed a norm. If $\rho\in X$ and $\lambda\in \R$, then $\|\lambda\rho\|_\Xi' = |\lambda|\|\rho\|_\Xi'$ for the same reason as for $\|\lambda\gamma\|_X=|\lambda|\|\gamma\|_X$ in the proof of Lemma \ref{Lemma: The X-norm}. It remains to demonstrate the triangle inequality. Let $\rho,\tilde\rho\in X$. Then, \begin{align*}\|\rho + \tilde\rho\|_{\Xi}':=& \inf \left\{\|\sqrt{\rho_+}\|_{H^1}^2 + \|\sqrt{\rho_-}\|_{H^1}^2 \middle| \rho +\tilde\rho= \rho^+ - \rho^-, \quad \rho^\pm\in\Xi^+\right\} \\
    \leq& \inf \bigg\{\|\sqrt{\rho_++\tilde\rho^+}\|_{H^1}^2 + \|\sqrt{\rho_-+\tilde\rho^-}\|_{H^1}^2 \; \bigg|\; \rho = \rho^+ - \rho^- ;\quad\tilde\rho = \tilde\rho^+ + \tilde\rho^- ;\quad \rho^\pm, \tilde\rho^\pm\in\Xi^+\bigg\}\\
    \leq& \inf \bigg\{\|\sqrt{\rho_+}\|_{H^1}^2 + \|\sqrt{\rho_-}\|_{H^1}^2 +\|\sqrt{\tilde\rho_+}\|_{H^1}^2 + \|\sqrt{\tilde\rho_-}\|_{H^1}^2 \;\bigg|\; \rho = \rho^+ - \rho^-; \\ &\hspace{8.5cm} 
    \tilde\rho = \tilde\rho^+ -\tilde\rho^-; \quad \rho^\pm, \tilde\rho^\pm\in\Xi^+\bigg\}\\ =& \|\rho\|_\Xi' + \|\tilde\rho\|_\Xi'.\end{align*}
The first inequality above comes from the fact that the infimum is taken over a smaller set, and the second inequality comes from the fact that 
\[\int (\nabla \sqrt{f+g})^2 \leq \int (\nabla\sqrt{f})^2 + \int (\nabla\sqrt{g})^2\] for any two positive functions $f$ and $g$, see \cite[pg. 256]{Lieb1983}.

We next show that $\|\cdot\|_\Xi'$ is equivalent to $\|\cdot\|_\Xi$. 
Let $\rho\in \Xi$, and pick a decomposition $\rho = \rho^+-\rho^-$, where $\rho^\pm$ are positive. Lemma \ref{Lemma: First Xi inequality} says that for any positive $\gamma^\pm\mapsto \rho^ \pm$, we have
\[\|\sqrt{\rho_+}\|_{H^1}^2 + \|\sqrt{\rho_-}\|_{H^1}^2\leq \Tr((1+T)\gamma^+) + \Tr((1+T)\gamma^-).\]
Thus, 
\[\|\sqrt{\rho_+}\|_{H^1}^2 + \|\sqrt{\rho_-}\|_{H^1}^2\leq \inf_{\gamma^+\in X^+, \; D(\gamma+)=\rho^+}\Tr((1+T)\gamma^+) + \inf_{\gamma^-\in X^+, \; D(\gamma^-)=\rho^-}\Tr((1+T)\gamma^-).\]
By Lemma \ref{Lemma: Xi norm equals infimum over decompositions.}, 
we get by taking the infimum over decompositions $\rho^\pm$ with $\rho= \rho^+-\rho^-$ that $\|\rho\|_\Xi'\leq \|\rho\|_\Xi$. 


Likewise, the second inequality in \eqref{Equation: norm inequality} implies that for any decomposition $\rho = \rho^+ - \rho^-$, with $\rho^\pm \geq 0$, we have
\[\|\rho\|_\Xi \leq \|\rho^+\|_\Xi+\|\rho^-\|_\Xi \leq (4\pi)^2 N\left( \|\sqrt{\rho_+}\|_{H^1}^2 + \|\sqrt{\rho_-}\|_{H^1}^2 \right).\] Taking the infimum over all decompositions, we get 
$\|\rho\|_\Xi\leq (4\pi)^2 N \|\rho\|_\Xi'$. 
\end{proof}
Corollary \ref{Corollary: Equivalent norms} means that $\Xi$ is a Banach space with respect to $\|\cdot\|_{\Xi}'$. Furthermore, Lemma \ref{Lemma: Xi-inequalities} implies that 
\begin{equation}\label{Equation: Characterization of Xi}
   \Xi \subset L^1_\R(\R^3)\cap L^3_\R(\R^3).
\end{equation}
and 
\[\Xi^+ = \left\{\rho\in L^1_\R(\R^3)\middle| \rho\geq 0,\; \nabla \sqrt{\rho}\in L^2(\R^3)\right\}.\]

The natural norm $\|\cdot\|_\Xi'$ is a refinement of the $L^1_\R(\R^3)\cap L^3_\R(\R^3)$-norm considered by Lieb in his analysis \cite{Lieb1983}. 

\begin{remark} While one might be tempted to think of $\|\rho\|_1 + \|\nabla \sqrt{|\rho|}\|_ 2^2$ as the correct norm for $\Xi$, this does not satisfy the triangle inequality, and is thus not a norm. For example, pick $\epsilon, \rho\in \Xi$, let $\rho_1 = \rho $ and $\rho_2 = \rho + \epsilon,$ and consider
\[(\nabla \sqrt{|\rho_1 - \rho_2|})^2= \frac{1}{2}\frac{(\nabla \epsilon)^2}{|\epsilon|}\]
while
\[(\nabla \sqrt{|\rho_1|})^2 + (\nabla\sqrt{|\rho_2|})^2=\frac{1}{2}\frac{(\nabla\rho)^2}{|\rho|} + \frac{1}{2}\frac{(\nabla(\rho+\epsilon))^2}{|\rho+\epsilon|}.\]
It is then not hard to find an $\epsilon $ such that 

\begin{align*}
    \int \left(\nabla\sqrt{|\rho_1-\rho_2|}\right)^2=\int  \frac{1}{2}\frac{ (\nabla \epsilon)^2}{|\epsilon|} > &\frac{1}{2} \int\left(\frac{(\nabla\rho)^2}{|\rho|} + \frac{(\nabla(\rho+\epsilon))^2}{|\rho+\epsilon|}\right) \\ =& \int \left(\nabla\sqrt{|\rho_1|}\right)^2+\int \left(\nabla\sqrt{|\rho_2|}\right)^2
\end{align*}
implying that triangle inequality does not hold. \end{remark}

Combining Lemma \ref{Lemma: D maps X+ surjectively onto Xi+} with \ref{Lemma: Xi-inequalities}, we see that 
\[D(\cD)= \left\{\rho\in L^1_\R(\R^3) \middle|\rho\geq 0 \text{ a.e.}, \int \rho = N, \int(\nabla \sqrt{\rho})^2 <+\infty\right\}.\]
This is exactly the domain $\mathscr{I}_N$ for the density functional defined by Lieb \cite{Lieb1983}. This means that our $F$ is exactly the functional considered by Lieb, which is also called Levy-Lieb functional \cite{lewin_universal_2023}. The set $D(\cD)$ is convex and closed in the $\Xi$-topology.

Since $\|D(\gamma)\|_\Xi\leq \|\gamma\|_X$ for all $\gamma\in X$, we have $\|D\|_{\cB(X,\Xi)} \leq 1$, so $D$ is bounded and surjective. This means that the adjoint $D^*:\Xi^* \to X^* $ is bounded and injective. Since $X/{\ker D}\cong \Xi$, the dual map $D^*$ is in fact an isometric isomorphism of $\Xi^*$ onto the annihilator of $\ker D$ \cite[Theorem 4.9]{Rudin1991},
\[\Xi^*\cong \ker D^\perp = \{v\in X^*\; |\;  v(\gamma)=0 \text{ for all $\gamma\in \ker D$}\}.\]
These are all the elements of $X^*$ that factor through $D$:
\[\Xi^* \cong\{ v\in X^*\; |\; \exists \tilde v: X\to \R: \; v=\tilde v\circ D\}.\]

By the inclusion in \eqref{Equation: Characterization of Xi}, the dual space satisfies 
\[L^\infty_\R(\R^3) + L^{3/2}_\R(\R^3) \subset \Xi^*,\]
 since $L^\infty_\R(\R^3) + L^{3/2}_\R(\R^3)$ is the dual space of $ L^1_\R(\R^3)\cap L^3_\R(\R^3)$, and is commonly called the Kato-Rellich class of potentials \cite[Definition 8.2]{Schmudgen2012}. The space $\Xi^*$ thus contains the Kato-Rellich class. However, only those $v\in \Xi^* $ which satisfy $D^*(v)\in \cR$ yield finite ground state energies, $E_\mathrm{D}(v)>-\infty$ and are thus valid as physical potentials. Since $D^*$ is bounded, the set $\cR^\Xi:=(D^*)^{-1}(\cR)$ is open. This is the set of all $v\in X^*$ with $-\nabla^2$-bound $<1$, i.e. there exist $0<a<1$ and $b\geq 0$ such that 
 \[\int vf^2 \leq a\int (\nabla f)^2 + b\int f^2\] for all $f\in L^2_\R(\R^3)$. 

 The set of infinitesimally $-\nabla^2$-bounded potentials $\cR_\epsilon^\Xi = (D^*)^{-1}(\cR_\epsilon)$ consists of all those potentials $v$ which have $-\nabla^2 $-bound $\epsilon $ for all $\epsilon>0$. It is well-known that 
\[L^\infty_\R(\R^3) + L^{3/2}_\R(\R^3) \subset (D^*)^{-1}(\cR_\epsilon),\]
i.e. that any Kato-Rellich potential has infinitesimal $-\nabla^2$-bound (see e.g. \cite{Simon1971} or \cite[Section 8.3]{Schmudgen2012}). This is essentially a consequence of Hölder's inequality $\int v\rho\leq \|v\|_{3/2}\|\rho\|_3$ and the Soblolev inequality $\|\rho\|_3\leq C\int (\nabla\sqrt{\rho})^2$. The more general Rollnik class of potentials is also contained in $\cR_\epsilon^\Xi$ \cite[Section 10.7]{Schmudgen2012} and we have the sequence of inclusions
\[\{\text{Kato-Rellich class}\}\subset\{\text{Rollnik class}\}\subset\cR_\epsilon^\Xi\subset \cR^\Xi \subset \Xi^*.\]
Since $\cR_\epsilon$ is closed in $X^*$ by Lemma \ref{Lemma: R inf is closed}, $\cR_\epsilon^\Xi$ is also closed and hence a Banach space in the $\Xi^*$-norm. In light of remark \ref{Remark: Positive norm conjecture}, we conjecture that the $\Xi^*$-norm is equivalent to 
\[\|v\|_{\Xi^*}':= \sup_{\rho\in \Xi^+} \frac{\int |v|\rho}{\int \rho + \int (\nabla\sqrt\rho)^2},\]
which would be a more convenient formulation of the dual norm.



\section[Appendix: Proof of (3.3)]{Appendix: Proof of \eqref{Equation: Super basic trace inequality}} \label{Appendix: Basic Trace Inequality}

Here, we provide a proof of \eqref{Equation: Super basic trace inequality}. Let $\varphi,\psi\in \cH$, and suppose that $\|\varphi\|=\|\psi\|=1$. If $\varphi=\psi$, then the inequality holds trivially. Suppose that $\varphi\neq\psi$, let $\{e_1,e_2\}$ denote an orthonormal basis for $\mathrm{span}_\C\{\varphi,\psi\}$, and assume without loss of generality that $e_1=\varphi$. We can then represent the operators $|\varphi\rangle\langle\varphi|$ and $|\psi\rangle\langle \psi|$ in this basis by
\[|\varphi\rangle\langle \varphi|=\begin{pmatrix}
    1 & 0 \\ 
    0 & 0
\end{pmatrix} \text{ and } |\psi\rangle\langle \psi|=\begin{pmatrix}
    |a|^2 & a\overline{b} \\
    b\overline{a} & |b|^2
\end{pmatrix},\] where $a= \langle \varphi,\psi\rangle$ and $b=\langle e_2,\psi\rangle $. We can find the value of the left hand side of \eqref{Equation: Super basic trace inequality} by calculating the eigenvalues of $|\varphi\rangle\langle\varphi|-|\psi\rangle\langle \psi|$. By using the above matrix representation and the fact that $|a|^2+|b|^2=1$, a straightforward calculation gives the eigenvalues $\pm |b|$ for $|\varphi\rangle\langle \varphi|-|\psi\rangle\langle\psi|$. We thus have
\[\Tr(\big||\varphi\rangle\langle\varphi|-|\psi\rangle\langle \psi|\big|) = 2|b|.\]
One can readily verify that
\[\|\psi-\varphi\|=\sqrt{|a-1|^2+|b|^2} \geq |b|,\] whence \eqref{Equation: Super basic trace inequality} follows.

\section*{Declarations}
\subsection*{Data availability statement}
No datasets were generated in the preparation of this
paper.
\subsection*{Conflict of interest} The authors have no relevant financial or non-financial interests
to disclose.

\bibliography{refs}

@article{Lieb1983,
	title = {{Density functionals for Coulomb systems}},
	volume = {24},
	copyright = {Copyright © 1983 John Wiley \& Sons, Inc.},
	issn = {1097-461X},
	doi = {10.1002/qua.560240302},
	language = {en},
	number = {3},
	urldate = {2024-12-15},
	journal = {Int. J. Quant. Chem.},
	author = {Lieb, Elliott H.},
	year = {1983},
	pages = {243--277},
}

@article{Kummer1967,
    author = {Kummer, Hans},
    title = {{n‐Representability Problem for Reduced Density Matrices}},
    journal = {J. Math. Phys.},
    volume = {8},
    number = {10},
    pages = {2063-2081},
    year = {1967},
    month = {10},
    issn = {0022-2488},
    doi = {10.1063/1.1705122}
}

@book{Schmudgen2012,
	address = {Dordrecht},
	series = {Graduate {Texts} in {Mathematics}},
	title = {Unbounded {Self}-adjoint {Operators} on {Hilbert} {Space}},
	volume = {265},
	copyright = {https://www.springernature.com/gp/researchers/text-and-data-mining},
	isbn = {978-94-007-4752-4},
	language = {en},
	urldate = {2024-12-18},
	publisher = {Springer Netherlands},
	author = {Schmüdgen, Konrad},
	year = {2012},
	doi = {10.1007/978-94-007-4753-1},
}

@article{Coleman1963,
  title = {Structure of Fermion Density Matrices},
  author = {Coleman, A. J.},
  journal = {Rev. Mod. Phys.},
  volume = {35},
  issue = {3},
  pages = {668--686},
  numpages = {0},
  year = {1963},
  month = {7},
  publisher = {American Physical Society},
  doi = {10.1103/RevModPhys.35.668},
  url = {https://link.aps.org/doi/10.1103/RevModPhys.35.668}
}

@book{RooijSchikhof1982,
  title = {A Second Course on Real Functions},
  author = {van Rooij, Arnoud C. M. and Schikhof, Wilhelmus Hendricus},
  year = {1982},
  publisher = {Cambridge University Press},
  address = {Cambridge},
  isbn = {978-0-521-23944-8 978-0-521-28361-8},
  langid = {english}
}

@article{GibneyBoynMazziotti2022,
author = {Gibney, Daniel and Boyn, Jan-Niklas and Mazziotti, David A.},
title = {Density Functional Theory Transformed into a One-Electron Reduced-Density-Matrix Functional Theory for the Capture of Static Correlation},
journal = {J. Phys. Chem. Lett.},
volume = {13},
number = {6},
pages = {1382-1388},
year = {2022},
doi = {10.1021/acs.jpclett.2c00083}
}

@book{BratteliVol2,
    title = {Operator algebras and quantum statistical mechanics : 2 : Equilibrium states : models in quantum statistical mechanics},
    author = {Ola Bratteli and D. W. Robinson},
    isbn = {3540614435},
    series = {International series of monographs on physics},
    year = {1997},
    publisher = {Springer},
    address={Berlin, Heidelberg}
}

@Article{Simon1971,
author={Simon, Barry},
title={Hamiltonians defined as quadratic forms},
journal={Commun. Math. Phys.},
year={1971},
month={9},
day={01},
volume={21},
number={3},
pages={192-210},
issn={1432-0916},
doi={10.1007/BF01647118},
url={https://doi.org/10.1007/BF01647118}
}

@article{gilbertHohenbergKohnTheoremNonlocal1975a,
  title = {{Hohenberg-Kohn Theorem for Nonlocal External Potentials}},
  author = {Gilbert, T. L.},
  year = {1975},
  month = sep,
  journal = {Phys. Rev. B},
  volume = {12},
  number = {6},
  pages = {2111--2120},
  publisher = {American Physical Society},
  doi = {10.1103/PhysRevB.12.2111}
}

@book{cioslowskiManyElectronDensitiesReduced2000,
  title = {Many-Electron Densities and Reduced Density Matrices},
  editor = {Cioslowski, Jerzy and Mezey, Paul G.},
  year = {2000},
  series = {{Mathematical and Computational Chemistry}},
  publisher = {Springer US},
  address = {Boston, MA},
  doi = {10.1007/978-1-4615-4211-7},
  urldate = {2025-03-28},
  copyright = {http://www.springer.com/tdm},
  isbn = {978-1-4613-6890-8},
  langid = {english}
}

@article{liebertRefiningRelatingFundamentals2023,
  title = {Refining and Relating Fundamentals of Functional Theory},
  author = {Liebert, Julia and Chaou, Adam Yanis and Schilling, Christian},
  year = {2023},
  month = jun,
  journal = {J. Chem. Phys.},
  volume = {158},
  number = {21},
  pages = {214108},
  issn = {0021-9606},
  doi = {10.1063/5.0143657}
}

@article{katoFundamentalPropertiesHamiltonian1951,
  title = {{Fundamental Properties of Hamiltonian Operators of Schr{\"o}dinger Type}},
  author = {Kato, Tosio},
  year = {1951},
  journal = {Trans. Am. Math. Soc.},
  volume = {70},
  number = {2},
  pages = {195--211},
  issn = {0002-9947, 1088-6850},
  doi = {10.1090/S0002-9947-1951-0041010-X},
  urldate = {2025-10-06},
  langid = {english}
}

@Article{Sherstnev2012,
author={Sherstnev, A. N.
and Tikhonov, O. E.},
title={On characterization of integrable sesquilinear forms},
journal={Math. Slovaca},
year={2012},
month={12},
day={01},
volume={62},
number={6},
pages={1167-1172},
issn={1337-2211},
doi={10.2478/s12175-012-0071-4},
url={https://doi.org/10.2478/s12175-012-0071-4}
}

@Article{Stolyarov2002,
author={Stolyarov, A. I.
and Tikhonov, O. E.
and Sherstnev, A. N.},
title={Characterization of Normal Traces on Von Neumann Algebras by Inequalities for the Modulus},
journal={Math. Notes},
year={2002},
month={9},
day={01},
volume={72},
number={3},
pages={411-416},
issn={1573-8876},
doi={10.1023/A:1020559623287},
url={https://doi.org/10.1023/A:1020559623287}
}

@article{Seidl1999,
  title = {Strictly correlated electrons in density-functional theory},
  author = {Seidl, Michael and Perdew, John P. and Levy, Mel},
  journal = {Phys. Rev. A},
  volume = {59},
  issue = {1},
  pages = {51--54},
  numpages = {0},
  year = {1999},
  month = {1},
  publisher = {American Physical Society},
  doi = {10.1103/PhysRevA.59.51},
  url = {https://link.aps.org/doi/10.1103/PhysRevA.59.51}
}

@misc{Seidl2017,
      title={The strictly-correlated electron functional for spherically symmetric systems revisited}, 
      author={Michael Seidl and Simone Di Marino and Augusto Gerolin and Luca Nenna and Klaas J. H. Giesbertz and Paola Gori-Giorgi},
      year={2017},
      eprint={1702.05022},
      archivePrefix={arXiv},
      primaryClass={cond-mat.str-el},
      url={https://arxiv.org/abs/1702.05022}, 
}

@article{CancesLewin2006,
    author = {Cancès, Eric and Stoltz, Gabriel and Lewin, Mathieu},
    title = {The electronic ground-state energy problem: A new reduced density matrix approach},
    journal = {J. Chem. Phys.},
    volume = {125},
    number = {6},
    pages = {064101},
    year = {2006},
    month = {08},
    issn = {0021-9606},
    doi = {10.1063/1.2222358},
    url = {https://doi.org/10.1063/1.2222358},
    eprint = {https://pubs.aip.org/aip/jcp/article-pdf/doi/10.1063/1.2222358/15388483/064101\_1\_online.pdf},
}

@article{Christian2018,
    author = {Schilling, Christian},
    title = {Communication: Relating the pure and ensemble density matrix functional},
    journal = {J. Chem. Phys.},
    volume = {149},
    number = {23},
    pages = {231102},
    year = {2018},
    month = {12},
    issn = {0021-9606},
    doi = {10.1063/1.5080088}
}

@article{PedersenTakesaki1973,
author = {Gert K. Pedersen and Masamichi Takesaki},
title = {{The Radon-Nikodym theorem for von Neumann algebras}},
volume = {130},
journal = {Acta Math.},
publisher = {Institut Mittag-Leffler},
pages = {53 -- 87},
year = {1973},
doi = {10.1007/BF02392262},
URL = {https://doi.org/10.1007/BF02392262}
}

@book{vanTielConvexAnalysis,
  title = {Convex Analysis: An Introductory Text},
  shorttitle = {Convex Analysis},
  author = {van Tiel, Jan},
  year = {1984},
  publisher = {Wiley},
  address = {Chichester [West Sussex]; New York},
  isbn = {978-0-471-90263-8 978-0-471-90265-2},
  langid = {english}
}

@book{Takesaki2003,
	address = {Berlin, Heidelberg},
	series = {Encyclopaedia of {Mathematical} {Sciences}},
	title = {Theory of {Operator} {Algebras} {II}},
	volume = {125},
	copyright = {http://www.springer.com/tdm},
	isbn = {978-3-642-07689-3},
	publisher = {Springer},
	author = {Takesaki, Masamichi},
	editor = {Cuntz, Joachim and Jones, Vaughan F. R.},
	year = {2003},
	doi = {10.1007/978-3-662-10451-4},
}

@book{Rudin1991,
  title = {Functional Analysis},
  author = {Rudin, Walter},
  year = {1991},
  series = {International Series in Pure and Applied Mathematics},
  edition = {2nd ed},
  publisher = {McGraw-Hill},
  address = {New York},
  isbn = {978-0-07-054236-5},
  langid = {english},
  lccn = {QA320 .R83 1991}
}

@book{Boyd_Vandenberghe_2004, place={Cambridge}, title={Convex Optimization}, publisher={Cambridge University Press}, author={Boyd, Stephen and Vandenberghe, Lieven}, year={2004},  address = {Cambridge}
}

@book{reed_methods_1980,
    address = {New York},
    title = {Methods of modern mathematical physics},
    volume = {1},
    isbn = {978-0-12-585050-6},
    language = {en},
    publisher = {Academic Press},
    author = {Reed, Michael and Simon, Barry},
    year = {1980},
}

@book{borwein_convex_2010,
    address = {Cambridge, UK},
    series = {Encyclopedia of mathematics and its applications},
    title = {Convex functions: constructions, characterizations and counterexamples},
    isbn = {978-0-521-85005-6 978-1-139-81142-2},
    shorttitle = {Convex functions},
    language = {en},
    number = {109},
    publisher = {Cambridge University Press},
    author = {Borwein, Jonathan M.},
    collaborator = {Vanderwerff, Jon D.},
    year = {2010},
}

@article{hohenberg_inhomogeneous_1964,
    title = {Inhomogeneous {Electron} {Gas}},
    volume = {136},
    url = {http://link.aps.org/doi/10.1103/PhysRev.136.B864},
    doi = {10.1103/PhysRev.136.B864},
    number = {3B},
    journal = {Phys. Rev.},
    author = {Hohenberg, P. and Kohn, W.},
    month = nov,
    year = {1964},
    pages = {B864--B871},
}

@article{kohn_self-consistent_1965,
    title = {Self-{Consistent} {Equations} {Including} {Exchange} and {Correlation} {Effects}},
    volume = {140},
    url = {http://link.aps.org/doi/10.1103/PhysRev.140.A1133},
    doi = {10.1103/PhysRev.140.A1133},
    number = {4A},
    journal = {Phys. Rev.},
    author = {Kohn, W. and Sham, L. J.},
    month = nov,
    year = {1965},
    pages = {A1133--A1138},
}

@article{garrigue_unique_2018,
  title = {{Unique Continuation for Many-Body Schr{\"o}dinger Operators and the Hohenberg-Kohn Theorem}},
  author = {Garrigue, Louis},
  year = {2018},
  month = sep,
  journal = {Math. Phys. Anal. Geom.},
  volume = {21},
  number = {3},
  pages = {27},
  issn = {1385-0172, 1572-9656},
  doi = {10.1007/s11040-018-9287-z},
  urldate = {2025-10-06},
  langid = {english}
}

@article{teale_dft_2022,
    title = {{DFT} exchange: sharing perspectives on the workhorse of quantum chemistry and materials science},
    issn = {1463-9076, 1463-9084},
    shorttitle = {{DFT} exchange},
    url = {http://xlink.rsc.org/?DOI=D2CP02827A},
    doi = {10.1039/D2CP02827A},
    language = {en},
    journal = {Phys. Chem. Chem. Phys.},
    author = {Teale, Andrew M. and Helgaker, Trygve and Savin, Andreas and {et al.}},
    year = {2022},
    volume = {24},
    pages = {28700-28781},
}

@incollection{toulouse_review_2023,
  title = {{Review of Approximations for the Exchange-Correlation Energy in Density-Functional Theory}},
  booktitle = {Density Functional Theory},
  author = {Toulouse, Julien},
  editor = {Canc{\`e}s, Eric and Friesecke, Gero},
  year = {2023},
  pages = {1--90},
  publisher = {Springer International Publishing},
  address = {Cham},
  doi = {10.1007/978-3-031-22340-2_1},
  urldate = {2025-10-06},
  isbn = {978-3-031-22339-6 978-3-031-22340-2},
  langid = {english}
}

@incollection{pernal_reduced_2015,
    title = {Reduced {Density} {Matrix} {Functional} {Theory} ({RDMFT}) and {Linear} {Response} {Time}-{Dependent} {RDMFT} ({TD}-{RDMFT})},
    isbn = {978-3-319-22081-9},
    language = {en},
    booktitle = {Density-{Functional} {Methods} for {Excited} {States}},
    publisher = {Springer},
    address = {Cham},
    author = {Pernal, Katarzyna and Giesbertz, Klaas J. H.},
    year = {2015},
    doi = {10.1007/128\_2015\_624},
    note = {ISSN: 1436-5049},
    pages = {125--183},
}

@incollection{lewin_universal_2023,
  title = {{Universal Functionals in Density Functional Theory}},
  booktitle = {Density Functional Theory},
  author = {Lewin, Mathieu and Lieb, Elliott H. and Seiringer, Robert},
  editor = {Canc{\`e}s, Eric and Friesecke, Gero},
  year = {2023},
  pages = {115--182},
  publisher = {Springer International Publishing},
  address = {Cham},
  doi = {10.1007/978-3-031-22340-2_3},
  urldate = {2025-10-06},
  isbn = {978-3-031-22339-6 978-3-031-22340-2},
  langid = {english}
}

@article{sutterSolutionVrepresentabilityProblem2024,
  title = {Solution of the V-Representability Problem on a One-Dimensional Torus},
  author = {Sutter, Sarina M and Penz, Markus and Ruggenthaler, Michael and van Leeuwen, Robert and Giesbertz, Klaas J H},
  year = 2024,
  month = nov,
  journal = {Journal of Physics A: Mathematical and Theoretical},
  volume = {57},
  number = {47},
  pages = {475202},
  publisher = {IOP Publishing},
  issn = {1751-8121},
  doi = {10.1088/1751-8121/ad8a2a},
  urldate = {2025-11-09},
  langid = {english}
}

\end{document}